\documentclass[preprint]{sigplanconf}
\usepackage{amsmath}
\usepackage{graphicx}

\newcommand{\true}{\mathit{true}}
\newcommand{\false}{\mathit{false}}
\newcommand{\varoast}{\star}

\newcommand{\Le}[1]{L(#1)}
\newcommand{\Lp}[1]{L_p(#1)}
\newcommand{\der}[2]{\ensuremath{\partial_{#1}(#2)}}
\newcommand{\nul}[1]{\ensuremath{\nu(#1)}}
\newcommand{\dnf}{\mathit{nf}}
\newcommand{\nf}{\mathit{nf}_{\!\mathbf{\epsilon}}}
\newcounter{item}
\newtheorem{theorem}[item]{Theorem}
\newtheorem{lemma}[item]{Lemma}
\newtheorem{definition}[item]{Definition}
\newtheorem{example}[item]{Example}
\newtheorem{remark}[item]{Remark}
\newtheorem{corollary}[item]{Corollary}
\newenvironment{proof}{\begin{trivlist}\item[]{\em Proof.}}{\end{trivlist}}

\begin{document}


\titlebanner{}
\preprintfooter{short description of paper}   

\title{Regular Expressions, {\em au point}\\}

\authorinfo{Andrea Asperti}
           {Department of Computer Science, \\
            University of Bologna\\
            Mura Anteo Zamboni 7, 40127, Bologna, ITALY}
           {asperti@cs.unibo.it}

\authorinfo{Claudio Sacerdoti Coen}
           {Department of Computer Science,\\
	   University of Bologna\\
            Mura Anteo Zamboni 7, 40127, Bologna, ITALY}
           {sacerdot@cs.unibo.it}

\authorinfo{Enrico Tassi}
           {Microsoft Research-INRIA Joint Center}
           {enrico.tassi@inria.fr}


\maketitle

\begin{abstract}
We introduce a new technique for constructing a finite state 
deterministic automaton from a regular expression, based on the idea 
of marking a suitable set of positions {\em inside} the expression,
intuitively representing the possible points reached after the 
processing of an initial prefix of the 
input string. {\em Pointed} regular expressions
join the elegance and the symbolic appealingness of Brzozowski's
derivatives, with the effectiveness of McNaughton and Yamada's
labelling technique, essentially combining the best of the two 
approaches. 
\end{abstract}



\section{Introduction}
There is hardly a subject in Theoretical Computer Science 
that, in view of its relevance and elegance, has been so thoroughly
investigated as the notion of {\em regular expression} and its relation
with {\em finite state automata} (see e.g. \cite{RS97,EllulKSW05} for
some recent surveys). All the studies in this area
have been traditionally inspired by two precursory, basilar works:
Brzozowski's theory of {\em derivatives} \cite{Brzozowski64}, and McNaughton and Yamada's
algorithm \cite{McNY60}. The main advantages of derivatives are that they are
syntactically appealing, easy to grasp and to prove correct (see 
\cite{OwensRT09} for a recent revisitation). 
On the other side, McNaughton and Yamada's approach
results in a particularly efficient algorithm, still used by
most pattern matchers like the popular grep and egrep utilities. 
The relation between the two approaches has been deeply investigated
too, starting from the seminal work by Berry and Sethi \cite{BS86} where it
is shown how to refine Brzozowski's method to get to the efficient
algorithm (Berry and Sethi' algorithm has been further improved by
later authors \cite{Bruggemann-Klein93,ChangP92}).

Regular expressions are such small world that it is much at no one's surprise
that all different approaches, at the end, turn out to be equivalent;
still, their philosophy, their underlying intuition, and the techniques
to be deployed can be sensibly different. Without having the
pretension to say anything really original on the subject, we introduce
in this paper a notion of {\em pointed} regular expression, that 
provides a cheap palliative for derivatives and allows a simple, direct 
and efficient construction of the deterministic finite automaton. 
Remarkably, the formal correspondence between pointed 
expressions and Brzozowski's derivatives is unexpectedly entangled 
(see Section~\ref{relation}) testifying the novelty and the not-so-trivial 
nature of the notion.


The idea of pointed expressions was suggested 
by an attempt of formalizing the theory
of regular languages by means of an interactive prover\footnote{The rule
of the game was to avoid overkilling, i.e. not make it more complex
than deserved.}. At first,
we started considering derivatives, since they looked more suitable
to the kind of symbolic manipulations that can be easily dealt with 
by means of these tools. However, the need to consider {\em sets} of
derivatives and, especially, to reason modulo associativity, commutativity
and idempotence of sum, prompted us to look for an alternative notion.
Now, it is clear that, in some sense, the derivative of a regular expression
$e$ is a set of  ``subexpressions'' of $e$\footnote{This is also the reason why, at
the end, we only have a finite number of derivatives.}: the only, 
crucial, difference is that we cannot forget their context. So, the natural 
solution is to {\em point} at subexpressions {\em inside} the original term.
This immediately leads to the notion of {\em pointed} regular expression
(pre), that is just a normal regular expression where some positions
(it is enough to consider individual characters) have been pointed
out. Intuitively, the points mark the positions inside the regular
expression which have been reached after reading some prefix of
the input string, or better the positions where the processing
of the remaining string has to be started. Each pointed expression
for $e$ represents a state of the {\em deterministic} automaton associated
with $e$; since we obviously have only a finite number of possible
labellings, the number of states of the automaton is finite.

Pointed regular expressions allow the {\em direct} construction of 
the DFA \cite{Kleene56}
associated with a regular expression, in a way that is
simple, intuitive, and efficient (the task is traditionally 
considered as {\em very involved} in the literature: see e.g
\cite{RS97}, pag.71).

In the imposing
bibliography on regular expressions - as far as we could discover -
the only author mentioning a notion close to ours is Watson \cite{Watson01,Watson02}. 
However, he only deals with single points, while the most 
interesting properties
of {\em pre} derive by their implicit additive nature (such as the
possibility to compute the {\em move} operation by a single pass on
the marked expression: see definition~\ref{move}).

\section{Regular expressions}

\begin{definition}
A regular expression over the alphabet $\Sigma$ is 
an expression $e$ generated by the following grammar:
\[E ::= \emptyset | \epsilon | a | E+E | EE | E^* \]
with $a \in \Sigma$
\end{definition}

\begin{definition}
The language $\Le{e}$ associated with the regular expression $e$ 
is defined by the following rules:
\[
\begin{array}{rcl}
\Le \emptyset &=& \emptyset\\
\Le \epsilon &=& \{\epsilon\}\\
\Le a &=& \{ a \}\\
\Le{e_1 + e_2} &=& \Le{e_1} \cup \Le{e_2}\\
\Le{e_1e_2} &=& \Le{e_1} \cdot \Le{e_2} \\ 
\Le{e^*} &=& \Le{e}^* \\ 
\end{array}
\]
where $\epsilon$ is the empty string, 
$L_1 \cdot L_2 = \{~ l_1l_2 ~|~ l_1 \in L_1,~ l_2 \in L_2 \}$ is the
concatenation of $L_1$ and $L_2$
and $L^*$ is the so called Kleene's closure of $L$: $L^* = \bigcup_{i=0}^\infty L^i$, with $L^0 = {\epsilon}$ and $L^{i+1} = L \cdot L^i$.
\end{definition}

\begin{definition}[nullable]~\\
A regular expression $e$ is said to be nullable if $\epsilon \in \Le{e}$.
\end{definition}

\noindent
The fact of being  nullable is decidable; it is easy to prove that the 
characteristic function $\nul{e}$ can be computed by the following 
rules:
\[
\begin{array}{rcl}
\nul{\emptyset} &=& \false \\
\nul{\epsilon} &=& \true \\
\nul{a} &=& \false \\
\nul{e_1+e_2} &=& \nul{e_1} \vee \nul{e_2}\\
\nul{e_1e_2} &=& \nul{e_1} \wedge \nul{e_2}\\
\nul{e^*} &=& \true\\
\end{array}
\]

\begin{definition}
A deterministic finite automaton (DFA) is a quintuple $(Q,\Sigma,q_0,t,F)$ 
where
\begin{itemize}
\item $Q$ is a finite set of states;
\item $\Sigma$ is the input alphabet;
\item $q_0 \in Q$ is the initial state;
\item $t: Q \times \Sigma \to Q$ is the state transition function;
\item $F \subseteq Q$ is the set of final states.
\end{itemize}
\end{definition}

\noindent
The transition function $t$ is extended to strings in the following way:

\begin{definition}Given a function $t:Q \times \Sigma \to Q$, 
the function $t^*:Q \times \Sigma^* \to Q$ is defined as follows:
\[t^*(q,w) = \begin{cases} t(q,\epsilon) = q \\ t(q,aw') = t^*(t(q,a),w')
\end{cases}\]
\end{definition}

\begin{definition}
Let $A=(Q,\Sigma,q_0,t,F)$ be a DFA; the language recognized 
$A$ is defined as follows:
\[L(A) = \{w | t^*(q_0,w) \in F \} \]
\end{definition}

\section{Pointed regular expressions}

\begin{definition}\ \label{def:pre}
\begin{enumerate}
\item A {\em pointed item} over the alphabet $\Sigma$ is 
an expression $e$ generated by following grammar:
\[E ::= \emptyset | \epsilon | a | \bullet a | E+E | EE | E^* \]
with $a \in \Sigma$;
\item A {\em pointed regular expression} (pre) is a pair 
$\langle e,b \rangle$ where $b$ is a boolean and $e$ is a 
pointed item.
\end{enumerate}
\end{definition}
The term $\bullet a$ is used to point to a position inside the regular
expression, preceding the given occurrence of $a$. 
In a pointed regular expression, 
the boolean must be intuitively understood as the possibility to have
a trailing point at the end of the expression. 

\begin{definition} \label{def:carrier}
The {\em carrier} $|e|$ of an item $e$
is the regular expression obtained from $e$ by removing all the points.
Similarly, the {\em carrier} of a pointed regular expression is the 
carrier of its item.
\end{definition}
In the sequel, we shall often use the same notation for functions
defined over items or pres, leaving to the reader the simple 
disambiguation task. Moreover, we use the notation $\epsilon(b)$, where $b$ 
is a boolean, with the following meaning:
\[\epsilon(\true)=\{\epsilon\}\quad\quad \epsilon(\false)=\emptyset\]

\begin{definition}\ \label{def:Lp}
\begin{enumerate}
\item The language $\Lp e$ associated with the item $e$ is defined by the
following rules:
\[
\begin{array}{rcl}
\Lp \emptyset &=& \emptyset\\
\Lp \epsilon &=& \emptyset\\
\Lp a &=& \emptyset\\
\Lp{\bullet a} &=& \{a\}\\
\Lp{e_1 + e_2} &=& \Lp{e_1} \cup \Lp{e_2}\\
\Lp{e_1e_2} &=& \Lp{e_1} \cdot \Le{|e_2|} \cup \Lp{e_2} \\ 
\Lp{e^*} &=& \Lp{e} \cdot \Le{|e|^*}  \\ 
\end{array}
\]
\item For a pointed regular expression $\langle e,b \rangle$ we define
\[
\Lp{\langle e,b \rangle} = \Lp e \cup \epsilon(b)\\
\]
\end{enumerate}
\end{definition}

\begin{example}
\label{ex:compact}
\[\Lp{(a+\bullet b)^*} = \Le{b(a+b)^*} \]
Indeed,
\[
\begin{array}{l}
\Lp{(a+\bullet b)^*} = \\
\quad = \Lp{a+\bullet b} \cdot \Le{|a+\bullet b|^*} \\
\quad = (\Lp{a} \cup \Lp{\bullet b}) \cdot \Le{(a+b)^*}\\
\quad = \{b\} \cdot \Le{(a+b)^*}\\
\quad = \Le{b(a+b)^*}\\
\end{array}
\]
\end{example}
Let us incidentally observe that, as shown by the previous
example, pointed regular expressions can provide a more compact 
syntax for denoting languages than traditional regular expressions. 
This may have important applications to 
the investigation
of the descriptional complexity (succinctness) of regular languages
(see e.g. \cite{Gelade10, GruberH08, HolzerK09}). 


\begin{example}
If $e$ contains no point (i.e. $e = |e|$) then $\Lp e = \emptyset$
\end{example}

\begin{lemma}
\label{lemma:epsilon}
If $e$ is a pointed item then $\epsilon \not\in \Lp e$. 
Hence, $\epsilon \in \Lp{\langle e, b \rangle}$ if and only if 
$b = \true$.
\end{lemma}
\begin{proof}
A trivial structural induction on $e$.
\end{proof}

\subsection{Broadcasting points}
Intuitively, a regular expression $e$ must be understood as a pointed
expression with a single point in front of it. Since however we only
allow points over symbols, we must broadcast this initial point
inside the expression, 
that essentially corresponds to the $\epsilon$-closure operation on
automata. We use the notation $\bullet(\cdot)$ to denote such an operation.

The broadcasting operator is also required to lift the item constructors 
(choice, concatenation and Kleene's star) from items to pres: for example,
to concatenate a pre $\langle e_1, \true \rangle$ 
with another pre $\langle e_2, b_2 \rangle$, we must first
broadcast the trailing point of the first expression inside $e_2$ and then
pre-pend $e_1$; similarly for the star operation.
We could define first the broadcasting
function $\bullet(\cdot)$ and then the lifted constructors; however, 
both the definition and the theory of the 
broadcasting function are simplified by making it
co-recursive with the lifted constructors.

\begin{definition}\ \label{def:bullet}
\begin{enumerate}
\item The function $\bullet(\cdot)$ from pointed item to pres is defined as
follows:
\[
\begin{array}{rcl}
\bullet(\emptyset)&=& \langle \emptyset,\false \rangle\\
\bullet(\epsilon) &=& \langle \epsilon,\true\rangle\\
\bullet(a) &=& \langle \bullet a,\false\rangle\\
\bullet(\bullet a) &=& \langle \bullet a,\false\rangle\\
\bullet(e_1 + e_2) &=& \bullet(e_1) \oplus \bullet(e_2) \\
\bullet(e_1e_2) &=& \bullet(e_1) \odot \langle e_2, \false \rangle\\
\bullet(e^*) &=& \langle e'^*,\true \rangle \mbox{ where }
\bullet(e) = \langle e',b'\rangle\\ \\ 
\end{array}
\]
\item The lifted constructors are defined as follows
\[
\begin{array}{l@{=}l}
\langle e'_1,b_1' \rangle \oplus \langle e'_2,b_2' \rangle & \langle e_1 + e_2, b_1'\vee b_2' \rangle \\
\langle e'_1,b_1' \rangle \odot \langle e'_2,b_2' \rangle & \begin{cases} \langle e_1'e_2', b_2' \rangle & \mbox{ when $b_1' = \false$}
 \\ \langle e_1'e_2'',b_2' \vee b_2'' \rangle & \mbox{ when $b_1' = \true$} \\ \quad & \mbox{ and
  $\bullet(e_2') = \langle e_2'',b_2'' \rangle$} \end{cases} \\
\langle e',b' \rangle^\varoast &
 \begin{cases}
   \langle e'^*, \false \rangle & \mbox{ when $b' = \false$ } \\
   \langle e''^*, \true \rangle & \mbox{ when $b' = \true$ } \\
                               & \mbox{ and $\bullet(e') = \langle e'',b''\rangle$}
 \end{cases}
\end{array}
\]
\end{enumerate}
\end{definition}

\noindent 
The apparent complexity of the previous definition should not 
hide the extreme simplicity of the broadcasting operation: on a sum we
proceed in parallel; on a concatenation $e_1e_2$, we first work on
$e_1$ and in case we reach its end we pursue broadcasting inside
$e_2$; in case of $e^*$ we broadcast the point inside $e$ recalling
that we shall eventually have a trailing point. 

\begin{example}
Suppose to broadcast a point inside 
\[(a + \epsilon)(b^*a + b)b\]
We start working in parallel on the first 
occurrence of $a$ (where the point stops), and on $\epsilon$ that
gets traversed. We have hence reached the end of $a + \epsilon$ and
we must pursue broadcasting inside $(b^*a + b)b$. Again, we work
in parallel on the two additive subterms $b^*a$ and $b$; the first
point is allowed to both enter the star, and to traverse it,
stopping in front of $a$; the second point just stops in front of
$b$. No point reached that end of $b^*a + b$ hence no further 
propagation is possible.
In conclusion:
\[\bullet((a + \epsilon)(b^*a + b)b) =
(\bullet a + \epsilon)((\bullet b)^*\bullet a + \bullet b)b\]
\end{example}

\begin{definition} The broadcasting function is extended to pres 
in the obvious way:
\[\bullet(\langle e,b\rangle) = 
\langle e',b\vee b' \rangle \mbox{ where }
\bullet(e) = \langle e',b'\rangle\]
\end{definition}

\noindent
As we shall prove in Corollary \ref{nullable}, broadcasting an initial 
point may reach the end of an expression $e$ if and only if $e$ 
is nullable.


\noindent


\noindent
The following theorem characterizes the broadcasting function and also
shows that the semantics of the lifted constructors on
pres is coherent with the corresponding constructors on items.

\begin{theorem}
\label{theo:broadcast}~
\begin{enumerate}
\item $\Lp{\bullet e} = \Lp e \cup \Le{|e|}$.
\item $L_p(e_1 \oplus e_2) = L_p(e_1) \cup L_p(e_2)$
\item $L_p(e_1 \odot e_2) = L_p(e_1) \cdot L(|e_2|) \cup  L_p(e_2)$
\item $L_p(e^\varoast) = L_p(e) \cdot L(|e|)^*$
\end{enumerate}
\end{theorem}
We do first the proof of 2., followed by the simultaneous proof of 1. and 3.,
and we conclude with the proof of 4.
\begin{proof}[of 2.]
We need to prove $L_p(e_1 \oplus e_2) = L_p(e_1) \cup L_p(e_2)$.
\[\begin{array}{l}
L_p(\langle e_1',b_1'\rangle \oplus \langle e_2',b_2' \rangle) =\\
\quad = L_p(\langle e_1' + e_2', b_1' \vee b_2' \rangle)\\
\quad = L_p(e_1' + e_2') \cup \epsilon(b_1') \cup \epsilon(b_2')\\
\quad = L_p(e_1')\cup \epsilon(b_1')  \cup L_p(e_2') \cup \epsilon(b_2')\\
\quad = L_p(e_1) \cup L_p(e_2)
\end{array}
\]
\end{proof}
\begin{proof}[of 1. and 3.]
We prove 1. ($\Lp{\bullet e} = \Lp e \cup \Le{|e|}$) by induction on the structure
of $e$, assuming that 3. holds on terms structurally smaller than $e$.
\begin{itemize}
\item $\Lp{\bullet(\emptyset)} = \Lp{\langle \emptyset, \false \rangle} 
= \emptyset = \Lp{\emptyset} \cup \Le{|\emptyset|}$.
\item $\Lp{\bullet(\epsilon)} = \Lp{\langle \epsilon, \true \rangle} 
= \{\epsilon\} = \Lp{\epsilon} \cup \Lp{|\epsilon|}$.
\item $\Lp{\bullet(a)} = \Lp{\langle a, \false \rangle} 
= \{a\} = \Lp{a} \cup \Le{|a|}$.
\item $\Lp{\bullet(\bullet a)} = \Lp{\langle \bullet a, \false \rangle} 
= \{a\} = \Lp{\bullet a} \cup \Le{|\bullet a|}$.
\item Let $e = e_1 + e_2$. By induction hypothesis we know that 
\[\Lp{\bullet(e_i)} = \Lp{e_i} \cup \Le{|e_i|}\]
Thus, by 2., we have
\[\begin{array}{l}
L_p(\bullet(e_1+e_2)) =\\
\quad = L_p(\bullet(e_1) \oplus \bullet(e_2)) \\
\quad = L_p(\bullet(e_1)) \cup L_p(\bullet(e_2)) \\
\quad = \Lp{e_1} \cup \Le{|e_1|} \cup \Lp{e_2} \cup \Le{|e_2|} \\
\quad = \Lp{e_1 + e_2} \cup \Le{|e_1 + e_2|}
\end{array}
\]
\item Let $e=e_1e_2$. By induction hypothesis we know that 
\[\Lp{\bullet(e_i)} = \Lp{e_i} \cup \Le{|e_i|}\]
Thus, by 3. over the structurally smaller terms $e_1$ and $e_2$
\[\begin{array}{l}
L_p(\bullet(e_1e_2)) =\\
\quad = L_p(\bullet(e_1) \odot \langle e_2,\false \rangle) \\
\quad = L_p(\bullet(e_1)) \cdot \Le{|e_2|} \cup L_p(e_2) \\
\quad = (L_p(e_1) \cup L(|e_1|)) \cdot \Le{|e_2|} \cup L_p(e_2) \\
\quad = L_p(e_1) \cdot L(|e_2|) \cup L(|e_1|) \cdot L(|e_2|) \cup L_p(e_2)\\
\quad = L_p(e_1e_2) \cup L(|e_1e_2|)
\end{array}
\]
\item Let $e=e_1^*$. By induction hypothesis we know that 
\[\Lp{\bullet(e_1)} = \Lp{e_1'} \cup \epsilon(b_1') = \Lp{e_1} \cup \Le{|e_1|}\]
and in particular, since by Lemma \ref{lemma:epsilon} $\epsilon \not\in \Lp{e_1}$, 
\[\Lp{e_1'} = \Lp{e_1} \cup (\Le{|e_1|}\setminus \epsilon(b_1'))\]
Then,
\[\begin{array}{l}
L_p(\bullet(e_1^*)) =\\
\quad = L_p(\langle e_1'^*,true\rangle)\\
\quad = L_p(e_1'^*) \cup \epsilon\\
\quad = L_p(e_1')\Le{|e_1^*|} \cup \epsilon\\
\quad = (\Lp{e_1} \cup (\Le{|e_1|}\setminus \epsilon(b_1')))\Le{|e_1^*|}  \cup \epsilon\\
\quad = \Lp{e_1}\Le{|e_1^*|} \cup (\Le{|e_1|}\setminus \epsilon(b_1'))\Le{|e_1^*|} \cup \epsilon\\
\quad = \Lp{e_1}\Le{|e_1^*|} \cup \Le{|e_1^*|}\\
\quad = \Lp{e_1^*} \cup \Le{|e_1^*|}
\end{array}\]
\end{itemize}

\noindent
Having proved 1. for $e$ assuming that 3. holds on terms structurally smaller
than $e$, we now assume that 1. holds for $e_1$ and $e_2$ in order to
prove 3.:
$L_p(e_1 \odot e_2) = L_p(e_1) \cdot L(|e_2|) \cup  L_p(e_2)$

We distinguish the two cases of the definition of $\odot$:

$$\begin{array}{l}
L_p(\langle e_1',\false \rangle \odot \langle e_2',b_2' \rangle) = \\
\quad = L_p(\langle e_1'e_2',b_2' \rangle) \\
\quad = L_p(e_1'e_2') \cup \epsilon(b_2') \\
\quad = L_p(e_1')\cdot L(|e_2'|) \cup L_p(e_2') \cup \epsilon(b_2')\\
\quad = L_p(e_1)\cdot L(|e_2|) \cup L_p(e_2)\\ \\
L_p(\langle e_1',true \rangle \odot \langle e_2',b_2' \rangle) = \\
\quad = L_p(\langle e_1'e_2'',b_2'\vee b_2'' \rangle) \\
\quad = L_p(e_1'e_2'') \cup \epsilon(b_2') \cup \epsilon(b_2'') \\
\quad = L_p(e_1')\cdot L(|e_2''|) \cup L_p(e_2'') \cup \epsilon(b_2')
 \cup \epsilon(b_2'')\\
\quad = L_p(e'_1)\cdot L(|e_2''|) \cup L_p(e_2') \cup L(|e_2'|) \cup \epsilon(b_2')\\
\quad = (L_p(e'_1) \cup \epsilon(\true)) \cdot L(|e_2|) \cup L_p(e_2') \cup \epsilon(b_2')\\
\quad = L_p(e_1)\cdot L(|e_2|) \cup L_p(e_2)
\end{array}$$
\end{proof}
\begin{proof}[of 4.]
We need to prove
$L_p(e^\varoast) = L_p(e) \cdot L(|e|)^*$.
We distinguish the two cases of the definition of $\cdot^\varoast$:
$$\begin{array}{l}
L_p(\langle e',\false \rangle^\varoast) = \\
\quad = L_p(\langle e'^*,\false \rangle) \\
\quad = L_p(e'^*) \\
\quad = L_p(e')\cdot L(|e'|)^* \\
\quad = (L_p(e') \cup \epsilon(\false))\cdot L(|e'|)^* \\
\quad = L_p(e)\cdot L(|e|)^*\\\\
L_p(\langle e',true \rangle^\varoast) = \\
\quad = L_p(\langle e''^*,true \rangle) \cup \epsilon \\
\quad = L_p(e''^*) \cup \epsilon \\
\quad = L_p(e'')\cdot L(|e''|)^* \cup \epsilon \\
\quad = (L_p(e') \cup L(|e'|))\cdot L(|e''|)^* \cup \epsilon \\
\quad = L_p(e')\cdot L(|e''|) \cup L(|e'|)\cdot L(|e''|)^* \cup \epsilon \\
\quad = L_p(e')\cdot L(|e''|) \cup L(|e'|)^* \\
\quad = (L_p(e') \cup \epsilon(\true))\cdot L(|e''|) \\
\quad = L_p(e)\cdot L(|e|)^*
\end{array}$$
\end{proof}

\begin{corollary}
\label{corollary:Lpbullet}
For any regular expression $e$, $\Le e = \Lp{\bullet e}$.
\end{corollary}

Another important corollary is that an initial point reaches 
the end of a (pointed) expression
$e$ if and only if $e$ is able to generate the empty
string.

\begin{corollary}
\label{nullable}
$\bullet e = \langle e', \true \rangle$ if and only if
$\epsilon \in \Le{|e|}$.
\end{corollary}
\begin{proof}
By theorem \ref{theo:broadcast} we know
that $\Lp{\bullet e} =\Lp e \cup \Le{|e|}$. So, if 
$\epsilon \in \Lp{\bullet e}$, since by Lemma \ref{lemma:epsilon} 
$\epsilon \not\in \Lp e$, it must be $\epsilon \in \Le{|e|}$.
Conversely, if $\epsilon \in \Le{|e|}$ then $\epsilon \in 
\Lp{\bullet e}$; if
$\bullet e = \langle e', b \rangle$, this is possible only 
provided $b=true$.
\end{proof}

To conclude this section, let us prove the idempotence
of the $\bullet{(\cdot)}$ function (it will only be used
in Section \ref{merging}, and can be skipped at a first reading). 
To this aim we need a technical lemma whose straightforward proof by
case analysis is omitted.

%

\begin{lemma}
~\vspace{-0.45cm}
$$
\begin{array}{ll}
1. & \bullet(e_1 \oplus e_2) = \bullet(e_1) \oplus \bullet(e_2)\\
2. & \bullet(e_1 \odot e_2) = \bullet(e_1) \odot e_2
\end{array}
$$
\end{lemma}

\begin{theorem}\label{bullid}
 ~\quad$\bullet(\bullet(e)) = \bullet(e)$
\end{theorem}
\begin{proof}
The proof is by induction on $e$.
\begin{itemize}
\item $\bullet(\bullet(\emptyset)) = \bullet (\langle \emptyset,\false \rangle)
 = \langle \emptyset,\false \vee \false \rangle = \bullet (\emptyset)$
\item $\bullet(\bullet(\epsilon)) = \bullet(\langle \epsilon,\true \rangle)
= \langle \epsilon, \true \vee \true \rangle = \bullet(\epsilon)$
\item $\bullet(\bullet(a)) = \bullet(\langle \bullet a, \false \rangle) =
\langle \bullet a, \false \vee \false \rangle = \bullet(a)$
\item $\bullet(\bullet(\bullet a)) = \bullet(\langle \bullet a, \false \rangle) =
\langle \bullet a, \false \vee \false \rangle = \bullet(\bullet a)$
\item If $e$ is $e_1 + e_2$ then
$$\begin{array}{l}
\bullet(\bullet(e_1+e_2)) =
\bullet(\bullet(e_1) \oplus \bullet(e_2))
= \bullet(\bullet(e_1)) \oplus \bullet(\bullet(e_2)) =\\
\quad = \bullet(e_1) \oplus \bullet(e_2)
= \bullet(e_1+e_2)\\
\end{array}$$
\item If $e$ is $e_1 e_2$ then
$$\begin{array}{l}
\bullet(\bullet(e_1e_2)) =
\bullet(\bullet(e_1) \odot \langle e_2, \false \rangle)
\bullet(\bullet(e_1)) \odot \langle e_2, \false \rangle =\\
\quad = \bullet(e_1) \odot \langle e_2, \false \rangle
= \bullet(e_1e_2)\\
\end{array}$$
\item If $e$ is $e_1^*$, let $\bullet (e_1) = \langle e',b' \rangle$ and
let $\bullet (e') = \langle e'',b'' \rangle$.
By induction hypothesis,
$$\langle e',b' \rangle = \bullet (e_1) = \bullet (\bullet (e_1)) = \bullet (\langle e',b' \rangle) = \langle e'',b' \vee b'' \rangle$$
and thus $e' = e''$. Finally
$$\begin{array}{l}
\bullet(\bullet(e_1^*))
= \bullet (\langle e'^*,\true \rangle)
= {\langle e''}^*, \true \vee b'' \rangle
= \langle e'^*,\true \rangle= \\
\quad = \bullet(e_1^*)
\end{array}$$
\end{itemize}
\end{proof}

\subsection{The move operation}
We now define the move operation, that corresponds to the 
advancement of the state in response to the processing of
an input character $a$. The intuition is clear: we have to
look at points inside $e$ preceding the given character $a$,
let the point traverse the character, and broadcast it. 
All other points must be removed.

\begin{definition}\label{move}\
\begin{enumerate}
\item The function $move(e,a)$ taking in input a pointed item 
$e$, a character $a \in \Sigma$ and giving back a pointer regular
expression is defined as follow, by induction on the structure of
$e$: 
\[
\begin{array}{rcl}
move(\emptyset,a)&=& \langle \emptyset,\false \rangle\\
move(\epsilon,a) &=& \langle \epsilon,\false\rangle\\
move(b,a) &=& \langle b,\false\rangle\\
move(\bullet a,a) &=& \langle a,\true\rangle \\
move(\bullet b,a) &=& \langle b,\false\rangle \mbox{ if } b\neq a\\
move(e_1 + e_2,a) &=& move(e_1,a) \oplus move(e_2,a)\\
move(e_1e_2,a) &=& move(e_1,a) \odot move(e_2,a) \\
move(e^*,a) &=& move(e,a)^\varoast \\
\end{array}
\]
\item The move function is extended to pres by just ignoring the trailing
point: 
$\quad move(\langle e,b\rangle,a) = move(e,a)$
\end{enumerate}
\end{definition}

\begin{example}
Let us consider the pre 
$(\bullet a + \epsilon)((\bullet b)^*\bullet a + \bullet b) b$
and the two moves w.r.t. the characters
$a$ and $b$. For $a$, we have two possible positions (all other
points gets erased); the innermost point stops in front of the final $b$,
the other one broadcast inside  $(b^*a + b) b$, so 
\[
move(
(\bullet a + \epsilon)((\bullet b)^*\bullet a + \bullet b)b ,a)
=
\langle (a + \epsilon)((\bullet b)^*\bullet a + \bullet b)\bullet b, \false \rangle
\]
For $b$, we have two positions too. The innermost point still
stops in front of the final $b$, while the other point reaches the end
of $b^*$ and must go back through  $b^*a$:
\[
move(
(\bullet a + \epsilon)((\bullet b)^*\bullet a + \bullet b)\bullet b ,b)
=
\langle (a + \epsilon)((\bullet b)^*\bullet a + b)\bullet b, \false \rangle
\]
\end{example}

\begin{theorem} 
\label{theo:move}
For any pointed regular expression $e$ and string
$w$, 
\[w \in \Lp{move(e,a)} \Leftrightarrow aw \in  \Lp e \]
\end{theorem}

\begin{proof}
The proof is by induction on the structure of $e$. 
\begin{itemize}
\item if $e$ is atomic, and $e$ is not a pointed symbol, then
both $\Lp{move(e,a)}$ and $\Lp e$ are empty, and hence both
sides are false for any $w$;
\item if $e = \bullet a$ then 
$\Lp{move(\bullet a,a)} = \Lp{\langle a,true\rangle} = \{\epsilon\}$
and $\Lp{\bullet a} = \{a\}$; 
\item if $e = \bullet b$ with $b \neq a$ then 
$\Lp{move(\bullet b,a)} = \Lp{\langle b,\false\rangle} = \emptyset$
and $\Lp{\bullet b} = \{b\}$; hence for any string $w$, both sides 
are false; 
\item if $e = e_1 + e_2$
by induction hypothesis
$w \in \Lp{move(e_i,a)} \Leftrightarrow aw \in \Lp{e_i}$,
hence, 

\[
\begin{array}{l}
w \in \Lp{move(e_1 + e_2,a)} \Leftrightarrow\\ 
\quad \Leftrightarrow w \in \Lp{move(e_1,a) \oplus move(e_2,a)} \\
\quad \Leftrightarrow w \in \Lp{move(e_1,a)} \cup \Lp{move(e_2,a)} \\
\quad \Leftrightarrow (w \in \Lp{move(e_1,a)}) \vee (w \in \Lp{move(e_2,a)}) \\
\quad \Leftrightarrow (aw \in \Lp{e_1}) \vee (aw \in \Lp{e_2}) \\
\quad \Leftrightarrow aw \in \Lp{e_1} \cup \Lp{e_2} \\
\quad \Leftrightarrow aw \in \Lp{e_1+e_2}
\end{array}
\]
\item suppose $e = e_1e_2$,
by induction hypothesis
$w \in \Lp{move(e_i,a)} \Leftrightarrow aw \in \Lp{e_i}$,
hence, 
\[
\begin{array}{l}
w \in \Lp{move(e_1e_2,a)} \Leftrightarrow\\ 
\quad \Leftrightarrow w \in \Lp{move(e_1,a) \odot move(e_2,a)} \\
\quad \Leftrightarrow w \in \Lp{move(e_1,a)} \cdot L|e_2| \cup \Lp{move(e_2,a)} \\
\quad \Leftrightarrow w \in \Lp{move(e_1,a)} \cdot L|e_2| \vee w \in \Lp{move(e_2,a)} \\
\quad \Leftrightarrow (\exists w_1,w_2, w=w_1w_2 \wedge w_1 \in \Lp{move(e_1,a)}\\ 
\quad\quad \wedge w_2 \in \Le{|e_2|}) \vee w \in \Lp{move(e_2,a)}\\
\quad \Leftrightarrow (\exists w_1,w_2, w=w_1w_2 \wedge aw_1 \in \Lp{e}\\ 
\quad\quad \wedge w_2 \in \Le{|e_2|}) \vee aw \in \Lp{e_2}\\
\quad \Leftrightarrow (aw \in \Lp{e_1} \cdot L|e_2|) \vee (aw \in \Lp{e_2}) \\
\quad \Leftrightarrow aw \in \Lp{e_1} \cdot L|e_2| \cup \in \Lp{e_2} \\
\quad \Leftrightarrow aw \in \Lp{e_1e_2} \\

\end{array}
\]

\item suppose $e = e_1^*$,
by induction hypothesis
$w \in \Lp{move(e_1,a)} \Leftrightarrow aw \in \Lp{e_1}$,
hence, 
\[
\begin{array}{l}
w \in \Lp{move(e_1^*,a)} \Leftrightarrow\\ 
\quad \Leftrightarrow w \in L_p(move(e_1,a))^\varoast\\
\quad \Leftrightarrow w \in L_p(move(e_1,a)) \cdot L(|move(e_1,a)|)^*\\
\quad \Leftrightarrow \exists w_1,w_2, w=w_1w_2 \wedge w_1 \in \Lp{move(e_1,a)}\\ 
\quad\quad \wedge w_2 \in \Le{|e_1|}^*\\
\quad \Leftrightarrow \exists w_1,w_2, w=w_1w_2 \wedge aw_1 \in \Lp{e_1}
 \wedge w_2 \in \Le{|e_1|}^*\\
\quad \Leftrightarrow aw \in L_p(e_1) \cdot L(|e_1|)^*\\
\quad \Leftrightarrow aw \in L_p(e_1^*)
\end{array}
\]
\end{itemize}
\end{proof}

We extend the move operations to strings as usual.
\begin{definition}
\[\begin{array}{l}
 move^*(e,\epsilon) = e \quad \quad
 move^*(e,aw) = move^*(move(e,a),w)
\end{array}\]
\end{definition}

\begin{theorem} 
\label{theo:move*}
For any pointed regular expression $e$ and all strings
$\alpha, \beta$, 
\[\beta \in \Lp{move^*(e,\alpha)} \Leftrightarrow \alpha\beta \in  \Lp e \]
\end{theorem}
\begin{proof}
A trivial induction on the length of $\alpha$, using theorem \ref{theo:move}.
\end{proof}

\begin{corollary}
For any pointed regular expression $e$ and any string
$\alpha$, 
\[\alpha \in  \Lp e \Leftrightarrow \exists e', \Lp{move^*(e,\alpha)} = 
\langle e',true \rangle \]
\end{corollary}
\begin{proof}
By Theorems \ref{theo:move*} and Lemma \ref{lemma:epsilon}.
\end{proof}

\subsection{From regular expressions to DFAs}

\begin{definition}
To any regular expression $e$ we may associate a DFA 
$D_e = (Q,\Sigma,q_0,t,F)$ 
defined in the following way:
\begin{itemize}
\item $Q$ is the set of all possible pointed expressions having $e$ as
carrier;
\item $\Sigma$ is the alphabet of the regular expression
\item $q_0$ is $\bullet e$;
\item $t$ is the move operation of definition \ref{move};
\item $F$ is the subset of pointed expressions $\langle e, b \rangle$ 
with $b = \true$.
\end{itemize}
\end{definition}

\begin{theorem} 
\label{theo:main}
$\quad L(D_e) = \Le e$
\end{theorem}
 \begin{proof}
By definition, 
\[w \in L(D_e) \leftrightarrow move^*(\bullet(e), w) = \langle e',true \rangle\]
for some $e'$. By the previous theorem, this is possible if an only if
$w \in \Lp{\bullet(e)}$, and by corollary \ref{corollary:Lpbullet}, 
$\Lp{\bullet(e)} = \Le e$.
\end{proof}

\begin{remark}\label{card}The fact that the set $Q$ of states of $D_e$ 
is finite is obvious: its cardinality is at most $2^{n+1}$ where $n$ is 
the number of symbols in $e$. This is one of the advantages
of pointed regular expressions w.r.t. derivatives, whose finite
nature only holds after a suitable quotient, and is a relatively
complex property to prove (see~\cite{Brzozowski64}).
\end{remark} 

\noindent
The automaton $D_e$ just defined may have many inaccessible states. We can
provide another algorithmic and direct construction that yields the same
automaton restricted to the accessible states only.

\begin{definition}\label{algorithmic1}
Let $e$ be a regular expression and let $q_0$ be $\bullet e$.
Let also
 $$\begin{array}{ll}Q_0 & := \{ q_0 \}\\
   Q_{n+1} & := Q_n \cup \{ e' | e' \not\in Q_n \land \exists a.\exists e \in Q_n. move(e,a) = e'\}\end{array}$$
 Since every $Q_n$ is a subset of the finite set of pointed regular expressions, there is an $m$ such that $Q_{m+1} = Q_m$.
We associate to $e$ the DFA $D_e = (Q_m,\Sigma,q_0,F,t)$ where $F$ and $t$
are defined as for the previous construction.
\end{definition}

\begin{figure}[htp]
\begin{center}
\includegraphics[width=.5\textwidth]{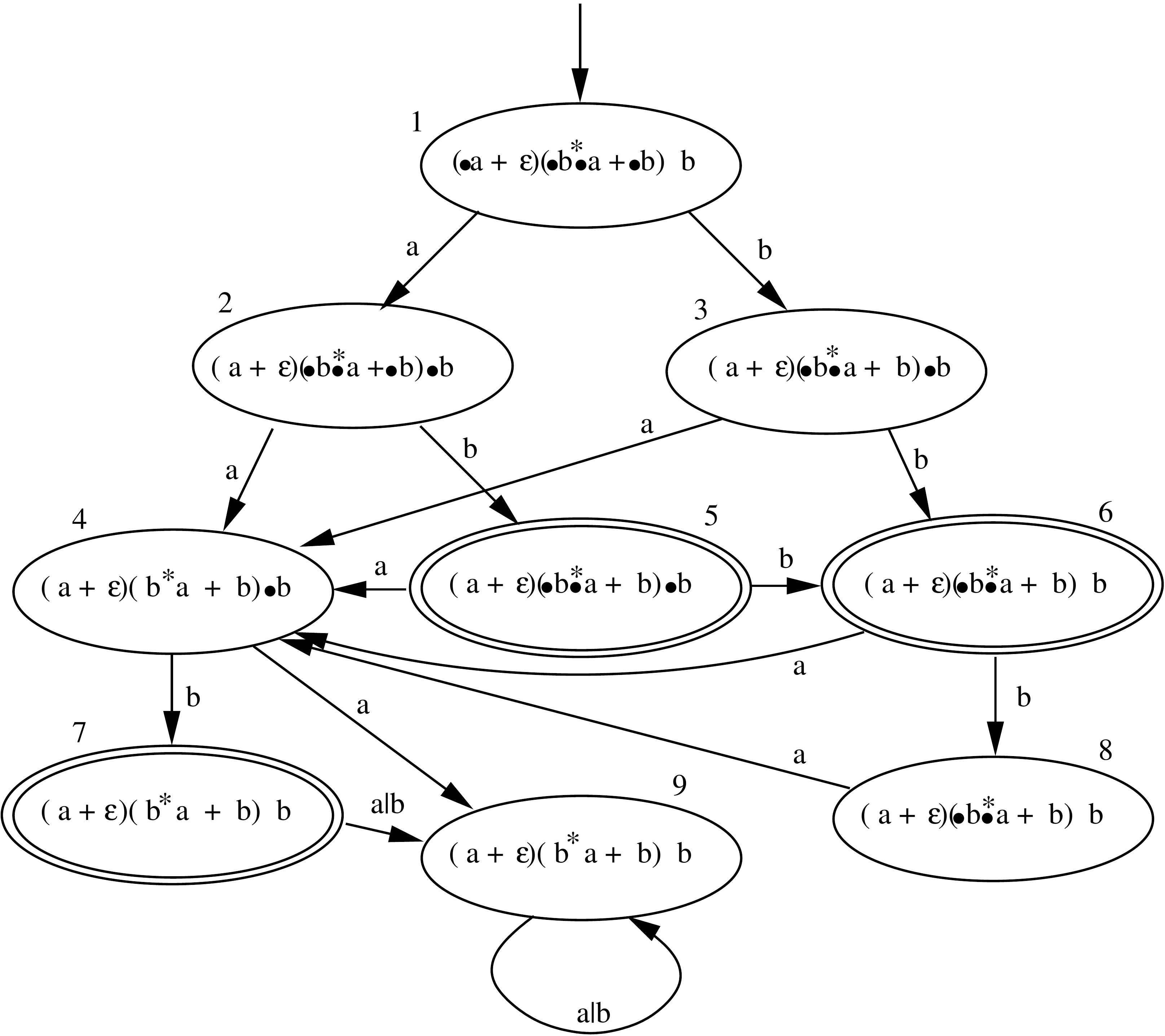}
\caption{DFA for $(a+\epsilon)(b^*a + b)b$\label{automaton}}
\end{center}
\end{figure}

In Figure \ref{automaton} we describe the DFA
associated with the regular expression $(a+\epsilon)(b^*a + b)b$.
The graphical description of the automaton is the traditional one,
with nodes for states and labelled arcs for transitions. 
Unreachable states are not shown.
Final states are emphasized by a double circle: since a state
$\langle e,b \rangle$ is final if and only if $b$ is true, we
may just label nodes with the item (for instance, the pair of
states $6-8$ and $7-9$ only differ for the fact that $6$ and $7$ 
are final, while $8$ and $9$ are not).

\subsection{Admissible relations and minimization}
The automaton in Figure \ref{automaton} is minimal. This is
not always the case. For instance, for the expression 
$(ac+bc)^*$ we obtain the automaton of Figure \ref{acUbc}, and 
it is easy to see that the two states corresponding to the
pres $(a\bullet c +bc)^*$ and $(ac+b\bullet c)^*$ are equivalent
(a way to prove it is to observe that they define the same
language). 

\begin{figure}[tp]
\begin{center}
\includegraphics[width=.5\textwidth]{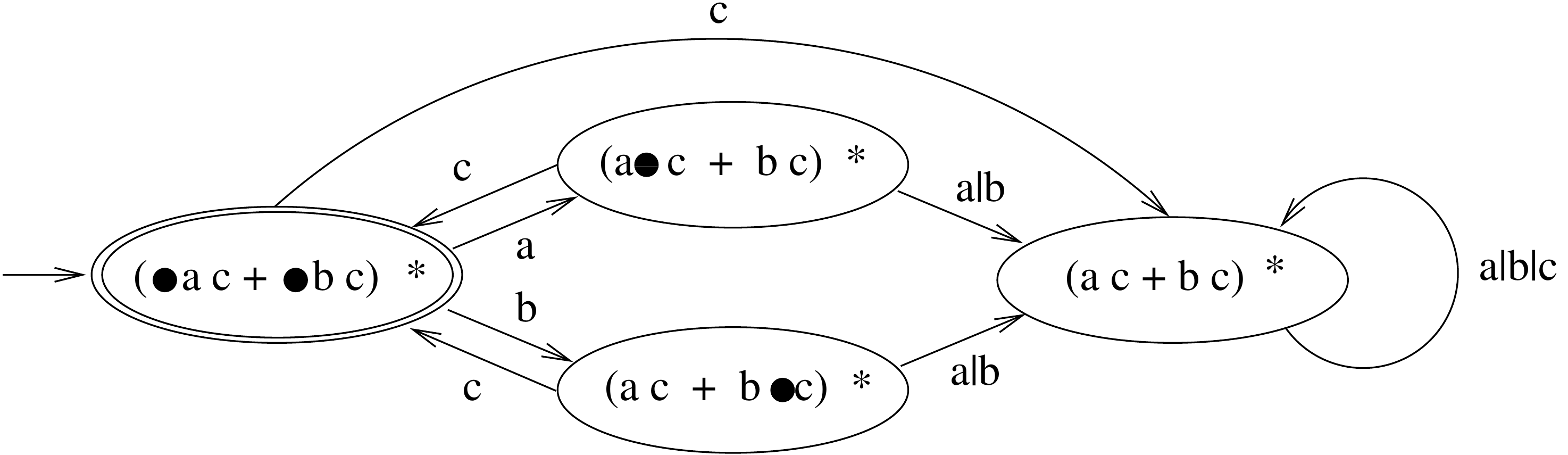}
\caption{DFA for $(ac+bc)^*$\label{acUbc}}
\end{center}
\end{figure}

%

\noindent
The latter remark, motivates the following definition.

\begin{definition}
An equivalence relation $\approx$ over pres having the same carrier is admissible when for all $e_1$ and $e_2$
\begin{itemize}
\item if $e_1 \approx e_2$ then $L_p(e_1) = L_p(e_2)$
\item if $e_1 \approx e_2$ then for all a $move(e_1,a) \approx move(e_2,a)$
\end{itemize}
\end{definition}

\begin{definition}\label{nonalgorithmic2}
To any regular expression $e$ and admissible equivalence relation over pres over
$e$, we can directly associate the DFA
$D_e/{\approx} = (Q/{\approx},\Sigma,[q_0]_\approx,move^*/{\approx},F/{\approx})$ where $move^*/{\approx}$ is the $move^*$ operation lifted to equivalence classes
thanks to the second admissibility condition.
\end{definition}

In place of working with equivalence classes, for formalization and implementation purposes it is simpler to work on representative of equivalence classes.
Instead of choosing a priori a representative of each equivalence class, we can slightly modify the algorithmic construction of definition~\ref{algorithmic1} so
that it dynamically identifies the representative of the equivalence classes.
It is sufficient to read each element of $Q_n$ as a representative of its
equivalence class and to change the test $e' \not \in Q_n$ so that the new
state $e'$ is compared to the representatives in $Q_n$ up to $\approx$:

\begin{definition}\label{algorithmic2}
In definition~\ref{algorithmic1} change the definition of $Q_{n+1}$ as follows:
 $$\begin{array}{l}
 Q_{n+1} :=\\
 ~~  Q_n \cup \{ e' | \exists a.\exists e \in Q_n. move(e,a) = e' \land \not \exists e'' \in Q_n. e' \approx e''\}
 \end{array}$$
The transition function $t$ is defined as $t(e,a)=e'$ where
$move(e,a) = e''$ and $e'$ is the unique state of $Q_m$ such that
$e' \equiv e''$.
\end{definition}
In an actual implementation, the transition function $t$ is computed together
with the sets $Q_n$ at no additional cost.

\begin{theorem}
Replacing each state $e$ of the automaton of definition~\ref{algorithmic2}
with $[e]/{\approx}$, we obtain the restriction of the automaton of
definition~\ref{nonalgorithmic2} to the accessible states.
\end{theorem}

We still need to prove that quotienting over $\approx$ does not change the
language recognized by the automaton.

\begin{theorem}
$\quad L(D_e/{\approx}) = L(e)$
\end{theorem}
\begin{proof}
By theorem~\ref{theo:main}, it is sufficient to prove
$L(D_e) = L(D_e/{\approx})$ or, equivalently, that for all $w$,
$move^*/{\approx}([q_0]/{\approx},w) \in F/{\approx} \iff
 move^*(q_0,w) \in F$. We show this to hold by proving
by induction over $w$ that for all $q$
$$[move^*(q,w)]/{\approx} = move^*/{\approx}([q]/{\approx},w)$$

\noindent
Base case:
$\begin{array}{l}
move^*/{\approx}([q]/{\approx},\epsilon) =
[q]/{\approx}
= [move^*(q,\epsilon)]/{\approx}
\end{array}
$

\noindent
Inductive step:
by condition (2) of admissibility, for all\\ $q_1 \in [q_0]/{\approx}$,
we have $move(q_1,a) \approx move(q_0,a)$ and thus
$$move/{\approx}([q_0]/{\approx},a) = [move(q_0,a)]/{\approx}$$
$$
\mbox{Hence }
\begin{array}[t]{l}
move^*/{\approx}([q_0]/{\approx},aw) =\\
\quad = move^*/{\approx}(move/{\approx}([q_0]/{\approx},a),w)\\
\quad = move^*/{\approx}([move(q_0,a)]/{\approx},w)\\
\quad = [move^*(move(q_0,a),w)]/{\approx}\\
\quad = [move^*(q_0,aw)]/{\approx}\\
\end{array}
$$
\end{proof}

The set of admissible equivalence relations over $e$ is a bounded lattice,
ordered by refinement, whose bottom element is syntactic identity and whose
top element is $e_1 \approx e_2$ iff $L(e_1) = L(e_2)$. Moreover, if
$\approx_1 < \approx_2$ (the first relation is a strict refinement of the
second one), the number of states of $D_e/{\approx_1}$ is strictly larger
than the number of states of $D_e/{\approx_2}$.

\begin{theorem}
If $\approx$ is the top element of the lattice, than $D_e/{\approx}$ is the
minimal automaton that recognizes $L(e)$.
\end{theorem}
\begin{proof}
By the previous theorem, $D_e/{\approx}$ recognizes $L(e)$ and has no
unreachable states. By absurd, let
$D' = (Q',\Sigma',q'_0,t',F')$ be another smaller automaton that recognizes
$L(e)$. Since the two automata are different, recognize the same languages and
have no unreachable states, there exists two words $w_1,w_2$ such
$t'(q'_0,w_1) = t'(q'_0,w_2)$ but
$[e_1]/{\approx} = move^*/{\approx}([q_0]/{\approx},w_1) \neq
 move^*/{\approx}([q_0]/{\approx},w_2) = [e_2]/{\approx}$ where $e_1$ and $e_2$
are any two representatives of their equivalence classes and thus $e_1 \not \approx e_2$. By definition of $\approx$, $L_p(e_1) \neq L_p(e_2)$. Without loss
of generality, let $w_3 \in L_p(e_1) \setminus L_p(e_2)$. We have
$w_1 w_3 \in L(e)$ and $w_2 w_3 \not \in L(e)$ because $D_e/{\approx}$ recognizes $L(e)$, which is absurd since
$t'(q'_0,w_1 w_3) = t'(q'_0,w_2 w_3)$ and $D'$ also recognizes $L(e)$.
\end{proof}

The previous theorem tells us that it is possible to associate to each state
of an automaton for $e$ (and in particular to the minimal automaton) 
a pre $e'$ over $e$ so that the language recognized by the automaton 
in the state $e'$ is $L_p(e')$, that provides a very suggestive labelling
of states. 


The characterization of the minimal automaton we just gave does not
seem to entail an original algorithmic construction, since does
not suggest any new effective way for computing $\approx$. However,
similarly to what has been done for derivatives (where we have similar
problems), it is interesting to investigate
admissible relations that are easier to compute and tend 
to produce small automata in most practical cases. 
In particular, in the next section, we shall investigate one important
relation providing a common quotient between the automata built
with pres and with Brzozowski's derivatives.

\section{Read back}
Intuitively, a pointed regular expression corresponds to a set
of regular expressions. In this section we shall formally investigate
this ``read back'' function; this will allow us to establish a more
syntactic relation between traditional regular expressions 
and their pointed version, and to compare our technique for building 
a DFA with that based on derivatives.

In the following sections we shall frequently deal with {\em sets} of
regular expressions (to be understood additively), that we prefer to 
the treatment of regular expressions up to associativity, 
commutativity and idempotence of the sum (ACI) that is for instance typical
of the traditional theory of derivatives (this also clarifies that 
ACI-rewriting is only used at the top level). 

It is hence useful to extend some syntactic operations, and 
especially concatenation, to sets of regular expressions, with 
the usual distributive meaning: 
if $e$ is a regular expression and $S$ is a set of regular 
expressions, then 
\[Se = \{e'e | e' \in S\} \]
We define $eS$ and $S_1S_2$ in a similar way. Moreover, every function
on regular expressions is implicitly lifted to sets of regular expressions by
taking its image. For example, \[ L(S) = \bigcup_{e \in S} L(e) \]

\begin{definition}
We associate to each item $e$ a set of regular expressions 
$R(e)$ defined by the
following rules:
\[
\begin{array}{rcl}
R(\emptyset) &=& \emptyset\\
R(\epsilon) &=& \emptyset\\
R(a) &=& \emptyset\\
R(\bullet a) &=& \{a\}\\
R(e_1 + e_2) &=& R(e_1) \cup R(e_2)\\
R(e_1e_2) &=& R(e_1)|e_2| \cup R(e_2) \\ 
R(e^*) &=& R(e)|e|^*  \\ 
\end{array}
\]
$R$ is extended to a pointed regular expression $\langle e,b \rangle$ as
follows
\[R(\langle e,b \rangle) = R(e) \cup \epsilon(b)\]
\end{definition}
Note that, for any item $e$, no regular expression in $R(e)$ is nullable.

\begin{example}
Since
$
\bullet((a + \epsilon) b^*) = \langle (\bullet a + \epsilon) (\bullet b)^*,\true \rangle
$
we have
$
R(\bullet((a + \epsilon) b^*)) = \{ a b^*, bb^*, \epsilon\}
$
\end{example}

The parallel between the syntactic read-back function $R$ and the semantics
$L_p$ of definition~\ref{def:Lp} is clear by inspection of the rules. Hence
the following lemma can be proved by a trivial induction over $e$.

\begin{lemma}\label{techx}
$\quad L(R(e)) = L_p(e)$
\end{lemma}

\begin{corollary}
For any regular expression $e$,
$L(R(\bullet (e))) = L(e)$
\end{corollary}

The previous corollary states that $R$ and $\bullet(\cdot)$ are semantically
inverse functions. Syntactically, they associate to each expression $e$
an interesting ``look-ahead'' normal form, constituted (up to associativity
of concatenation) by a set of expressions of the kind $ae_a$ 
(plus $\epsilon$ if $e$ is nullable), where
$e_a$ is a derivative of $e$ w.r.t. $a$ (although syntactically different
from Brzozowski's derivatives, defined in the next section).
 
This look-ahead normal form ($\dnf$) has an interest in its own, and can be 
simply defined by structural induction over $e$.

\begin{definition}
\[
\begin{array}{l}
\dnf(\emptyset)= \emptyset\\
\dnf(\epsilon)= \emptyset\\
\dnf(a) = \{a\}\\
\dnf(e_1+e_2) = \dnf(e_1) \cup \dnf(e_2)\\
\dnf(e_1e_2) = \dnf(e_1)e_2 \mbox{ if } \nul{e_1}=\false\\
\dnf(e_1e_2) = \dnf(e_1)e_2 \cup \dnf(e_2) \mbox{ if } \nul{e_1}=\true\\
\dnf(e^*) = \dnf(e)e^*\\
\end{array}
\]
\end{definition}

\begin{remark}\label{not_nullable}
It is easy to prove that, for each $e$, 
the set $\dnf(e)$ is made, up to associativity of concatenation,
only of expressions of the form $a$ or $a e_a$. In particular
no expression in $\dnf(e)$ is nullable!
\end{remark}

\noindent
The previous remark motivates the following definition.

\begin{definition}
$\quad \nf(e) = \dnf(e) \cup \epsilon(\nul{|e|})$
\end{definition}
The main properties of $\nf$ are expressed by the following
two lemmas, whose simple proof is left to the reader.
\begin{lemma}\label{lemma:nf}
\[
\begin{array}{l}
\nf(\emptyset)= \emptyset\\
\nf(\epsilon)= \{\epsilon\}\\
\nf(a) = \{a\}\\
\nf(e_1+e_2) = \nf(e_1) \cup \nf(e_2)\\
\nf(e_1e_2) = \nf(e_1)e_2 \mbox{ if } \nul{e_1}=\false\\
\nf(e_1e_2) = \dnf(e_1)e_2 \cup \nf(e_2) \mbox{ if } \nul{e_1}=\true\\
\nf(e^*) = \dnf(e)e^* \cup \epsilon(\nul{e})\\
\end{array}
\]
\end{lemma}

\begin{theorem}\label{techz}
 $\quad L(e) = L(\nf(e)) $
\end{theorem}

\begin{theorem}
For any pointed regular expression $e$,
\[R(\bullet(e)) = \nf(|e|) \cup R(e) \]
\end{theorem}
\begin{proof}
Let $\bullet(e) = \langle e',b' \rangle$; then $\epsilon \in R(\bullet(e))$
iff $b'=true$, iff $\nul{|e|} = \true$. Hence the goal reduces to prove
that $R(e') = \dnf{|e|} \cup R(e)$. 
We proceed by induction on the structure of $e$. 
\begin{itemize}
\item $e = \emptyset$, $\bullet(\emptyset) = \langle \emptyset, \false \rangle$
and 
$R(\emptyset) = \emptyset = \dnf(\emptyset)$
\item $e = \epsilon$, $\bullet(\epsilon) = \langle \epsilon, \true \rangle$
and $R(\epsilon) = \emptyset = \dnf(\epsilon)$
\item $e = a$: $(\bullet(a))= \langle \bullet a, \false \rangle$ and
$R(\bullet a) = \{a\} = \dnf(a)$
\item $e = \bullet a$: $(\bullet(\bullet a))= \langle \bullet a, \false \rangle$ and
$R(\bullet a) = \{a\} = \dnf(a) = \dnf(|\bullet a|) = \dnf(|\bullet a|) \cup R(\bullet a)$
\item $e = e_1+e_2$: let $\bullet(e_1+e_2) = \langle e_1'+e_2',b\rangle$;
then
\[
\begin{array}{l}
R(e_1'+e_2') = \\
\quad = R(e_1')\cup R(e_2') \\
\quad = \dnf(|e_1|) \cup R(e_1) \cup \dnf(|e_2|) \cup R(e_2)\\
\quad = \dnf{|e_1+e_2|} \cup R(e_1+e_2)
\end{array}
\] 
\item $e = e_1e_2$.
Let $\bullet(e_i) = \langle e_i',b_i' \rangle$. If $b_1' = \false$ then
$\bullet(e_1e_2) = \langle e_i'e_2,\false \rangle$; moreover we
know that $e_1$ is not nullable. We have then:
\[
\begin{array}{l}
R(e_1'e_2) = \\
\quad= R(e_1')|e_2| \cup R(e_2) \\
\quad= (\dnf(|e_1|) \cup R(e_1))|e_2| \cup R(e_2)\\
\quad= (\dnf(|e_1|)|e_2| \cup R(e_1)|e_2| \cup R(e_2)\\
\quad=  \dnf(|e_1e_2|) \cup R(e_1e_2)\\
\end{array}
\] 
If $b_1' = \true$ then
$\bullet(e_1e_2) = \langle e_i'e_2',b_2' \rangle$; moreover we
know that $e_1$ is nullable.
\[
\begin{array}{l}
R(e_1'e_2')= \\
\quad= R(e_1')|e_2| \cup R(e_2') \\
\quad= (\dnf(|e_1| \cup R(e_1))|e_2| \cup \dnf(|e_2|)) \cup R(e_2)\\
\quad= \dnf(|e_1|)|e_2| \cup \dnf(e_2) \cup R(e_1)|e_2| \cup R(e_2)\\
\quad= (\dnf(|e_1e_2|)) \cup R(e_1e_2)\\
\end{array}
\] 
\item $e = e_1^*$. Let $\bullet(e_1) = \langle e_i',b_i' \rangle$;
then $\bullet(e_1^*) = \langle e_i'^*,true \rangle$;
\[
\begin{array}{l}
R(e_1'^*) = \\
\quad= R(e_1')|e_1|^* \\
\quad= (\dnf(e_1) \cup R(e_1))|e_1|^*\\
\quad= \dnf(e_1)|e_1|^* \cup R(e_1))|e_1|^*\\
\quad= \dnf(e_1^*) \cup R(e_1^*)\\
\end{array}
\] 
\end{itemize}
\end{proof}
\begin{corollary}\label{R-bullet}
For all regular expression $e$,
$
R(\bullet(e)) = \nf(e)
$
\end{corollary}

\noindent
To conclude this section, in analogy with what we did for the 
semantic function in Theorem~\ref{theo:broadcast}, we express the 
behaviour of $R$ in terms of the {\em lifted} algebraic 
constructors. This will be useful in Theorem \ref{techy}.

\begin{lemma}\ 
\begin{enumerate}
\item $R(e_1 \oplus e_2) = R(e_1) \cup R(e_2)$
\item $R(\langle e_1',\false \rangle \odot e_2) = R(e_1')|e_2| \cup R(e_2)$
\item $R(\langle e_1',true \rangle \odot e_2) = R(e_1')|e_2| \cup \nf(|e_2|) 
\cup R(e_2)$
\item $R(\langle e_1',\false \rangle^\varoast) = R(e_1')|e_1^*|$
\item $R(\langle e_1',true \rangle^\varoast) = R(e_1')|e_1^*| 
\cup \nf(|e_1^*|) $
\end{enumerate}
\end{lemma}
\begin{proof}
Let $e_i = \langle e_i',b_i' \rangle$: 
\begin{enumerate}
\item 
$
\begin{array}[t]{l}
R(e_1 \oplus e_2) = \\
\quad = R(\langle e_1',b_1' \rangle \oplus langle e_2',b_2' \rangle) =\\ 
\quad = R(\langle e_1'+e_2',b_1' \vee b_2' \rangle)\\
\quad = R(e_1'+e_2') \cup \epsilon(b_1' \vee b_2')\\
\quad = R(e_1') \cup R(e_2') \cup \epsilon(b_1)' \cup \epsilon(b_2')\\
\quad = R(e_1') \cup \epsilon(b_1') \cup R(e_2') \cup \epsilon(b_2')\\
\quad = R(e_1) \cup R(e_2)
\end{array}
$
\item 
$
\begin{array}[t]{l}
R(\langle e_1',\false \rangle \odot \langle e_2',b_2'\rangle) =\\
\quad = R(\langle e_1'e_2',b_2' \rangle)\\
\quad = R(e_1')|e_2| \cup R(e_2') \cup \epsilon(b_2')\\
\quad = R(e_1')|e_2| \cup R(e_2)
\end{array}
$
\item let $\bullet(e_2') = \langle e_2'',b_2''\rangle$
\[
\begin{array}{l}
R(\langle e_1',true \rangle \odot \langle e_2',b_2'\rangle) =\\
\quad = R(\langle e_1'e_2'',b_2' \vee b_2'' \rangle)\\
\quad = R(e_1')|e_2| \cup R(e_2'') \cup \epsilon(b_2'') \cup \epsilon(b_2')\\
\quad = R(e_1')|e_2| \cup R(\bullet(e_2'))\cup \epsilon(b_2')\\
\quad = (R(e_1')|e_2| \cup nf(|e_2|) \cup R(e_2') \cup \epsilon(b_2') \\
\quad = R(e_1')|e_2| \cup nf(|e_2|) \cup R(e_2)
\end{array}
\]
\item 
$
\begin{array}[t]{l}
R(\langle e_1',\false \rangle^\varoast ) =
R(\langle e_1'^*,\false \rangle)
= R(e_1'^*)
= R(e_1')|e_1^*|
\end{array}
$
\item let $\bullet(e_1') = \langle e_1'',b_1''\rangle$; then
$R(\bullet(e_1')) = R(e_1'') \cup \epsilon(b_1'') =
\nf(|e_1|) \cup R(e_1')$, and $R(e_1'') = \dnf(|e_1|) \cup R(e_1')$. 
\[
\begin{array}{l}
R(\langle e_1',true \rangle^\varoast ) =\\
\quad = R(\langle e_1''^*,true \rangle)\\
\quad = R(e_1'')|e_1^*| \cup \epsilon(\true) \\
\quad = (R(e_1') \cup dnf(|e_1|))|e_1^*| \cup \epsilon(\true)\\
\quad = R(e_1')|e_1^*| \cup dnf(|e_1|)|e_1^*| \cup \epsilon(\true)\\
\quad = R(e_1')|e_1^*| \cup nf(|e_1^*|) 
\end{array}
\]
\end{enumerate}
\end{proof}

\subsection{Relation with Brzozowski's Derivatives}
\label{relation}
We are now ready to formally investigate the relation
between pointed expressions and Brzozowski's derivatives.
As we shall see, they give rise to quite different constructions
and the relation is less obvious than expected.\\
Let's start with recalling the formal definition.
\begin{definition}
$$
\begin{array}{rcl}
\der{a}{\emptyset} & = & \emptyset \\
\der{a}{\epsilon} & = & \emptyset \\
\der{a}{a} & = & \epsilon \\
\der{a}{b} & = & \emptyset \\
\der{a}{e_1 + e_2} &=& \der{a}{e_1} + \der{a}{e_2}\\
\der{a}{e_1e_2} &=& \der{a}{e_1}e_2 \mbox{ if not } $\nul{e_1}$\\
\der{a}{e_1e_2} &=& \der{a}{e_1}e_2 + \der{a}{e_2} \mbox{ if } $\nul{e_1}$\\\
\der{a}{e^*} &=& \der{a}{e}e^*\\
\end{array}
$$
\end{definition}
\begin{definition}
$$\begin{array}{rcl}
\der{\epsilon}{e} & = & e\\
\der{aw}{e} & = & \der{w}{\der{a}{e}}\\
\end{array}$$
\end{definition}

In general, given a regular expression $e$ over the alphabet $\Sigma$, 
the set $\{\der{w}{e} \;|\; w \in \Sigma^*\}$
of all its derivatives {\em is not} finite. In order to get a finite
set we must suitably quotient derivatives according to algebraic 
equalities between regular expressions. The choice of different 
set of equations gives rise
to different quotients, and hence to different automata. Since
for finiteness it is enough to consider associativity,
commutativity and idempotence of the sum (ACI), the traditional
theory of Brzozowski's derivatives is defined according to
these laws (although this is probably not the best choice from 
a practical point of view).

As a practical example, in Figure \ref{acUbc1} we describe 
the automata obtained using derivatives relative to the 
expression $(ac + bc)^*$ (compare it
with the automata of Figure \ref{acUbc}). Also, note that the
vertically aligned states are equivalent.

\begin{figure}[htp]
\begin{center}
\includegraphics[width=.5\textwidth]{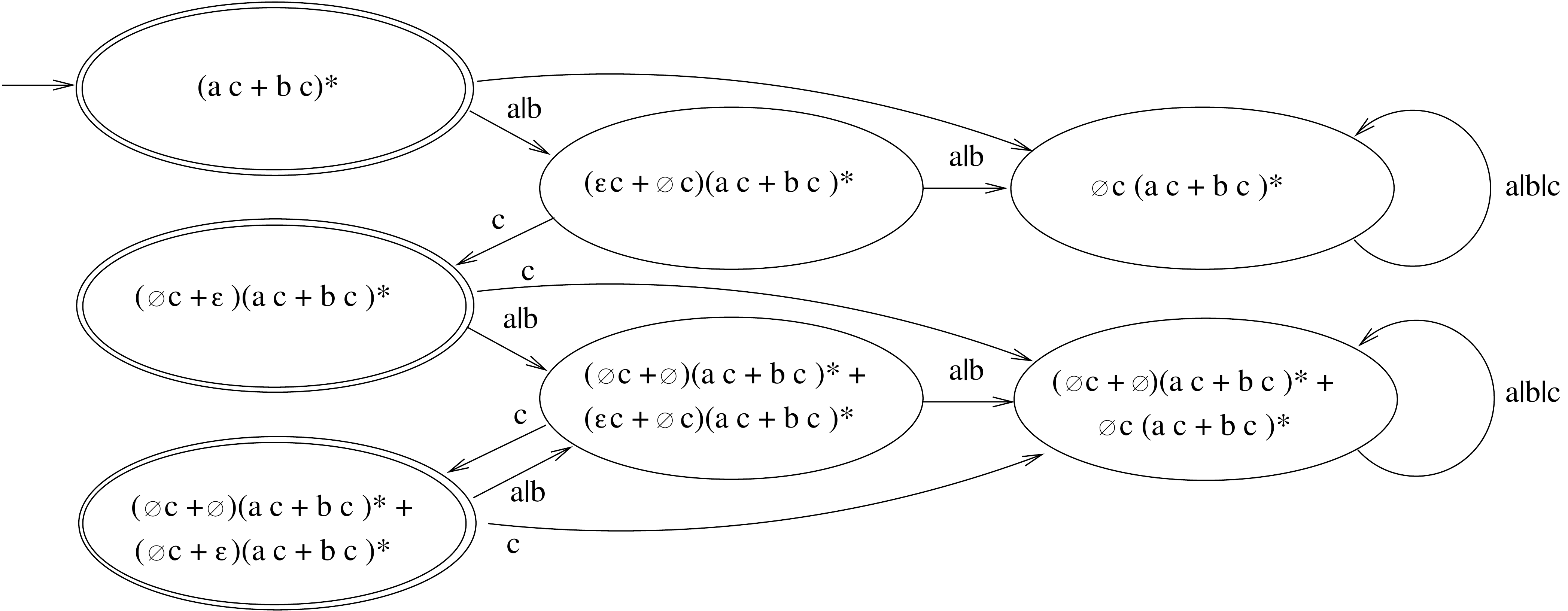}
\caption{Automaton with Brzozowski's derivatives\label{acUbc1}}
\end{center}
\end{figure}

\noindent
Let us remark, first of all, the heavy use of $ACI$. For instance 
\[\der{a}{(ac+bc)^*} = (\epsilon c + \emptyset c)(ac+bc)^*\]
while
\[\der{b}{(ac+bc)^*} = (\emptyset c + \epsilon c)(ac+bc)^*\] 
and they can be assimilated only up to commutativity of the sum. 
As another example,
\[
\begin{array}{l}
\der{a}{(\emptyset c + \emptyset)(ac+bc)^* + (\emptyset c + \epsilon)(ac+bc)^*} =\\
\quad=(\emptyset c + \emptyset)(ac+bc)^* + \\
\quad\quad((\emptyset c + \emptyset)(ac+bc)^* + (\epsilon c + \emptyset c)(ac+bc)^*)
\end{array}
\]
and the latter expression can be reduce to
\[(\emptyset c + \emptyset)(ac+bc)^* + (\epsilon c + \emptyset c)(ac+bc)^*)\] 
only using associativity and idempotence of the sum.

The second important remark is that, in general, it is not true that we may 
obtain the pre-automata by quotienting the derivative one (nor the other way
round). For instance, from the initial state, the two arcs labelled $a$ and
$b$ lead to a single state in the automata of Figure~\ref{acUbc1}, but in different
states in the automata of Figure~\ref{acUbc}. 

A natural question is hence to
understand if there exists a common {\em algebraic} quotient between the two 
constructions (not exploiting minimization). 

As we shall see, this can be achieved by identifying states with a same
readback in the case of pres, and states with similar look-ahead normal
form in the case of derivatives.

For instance, in the case of the two automata of Figures~\ref{acUbc} and
\ref{acUbc1}, we would obtain the common quotient of Figure~\ref{quotient}.

\begin{figure}[htp]
\begin{center}
\includegraphics[width=.5\textwidth]{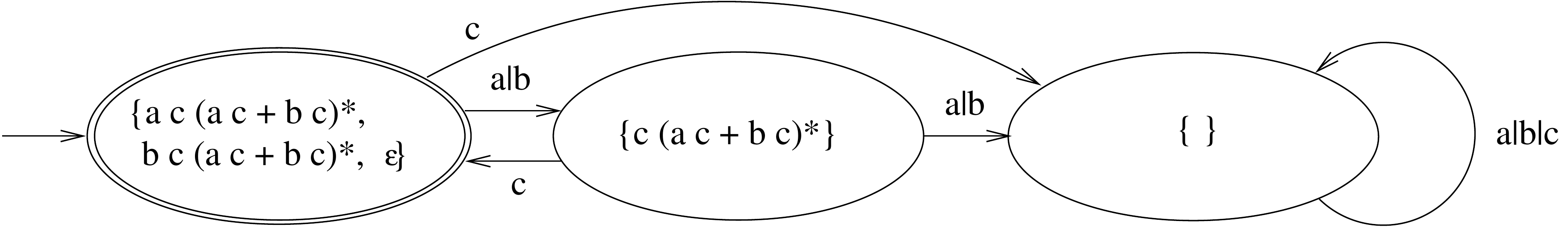}
\caption{A quotient of the two automatons\label{quotient}}
\end{center}
\end{figure}

The general picture is described by the commuting diagram
of Figure~\ref{commute}, whose proof will be the object of the 
next section (in Figure~\ref{commute}, $w$ obviously stands for
the string $a_1\dots a_n$).

\begin{figure}[htp]
\begin{center}
\includegraphics[width=.5\textwidth]{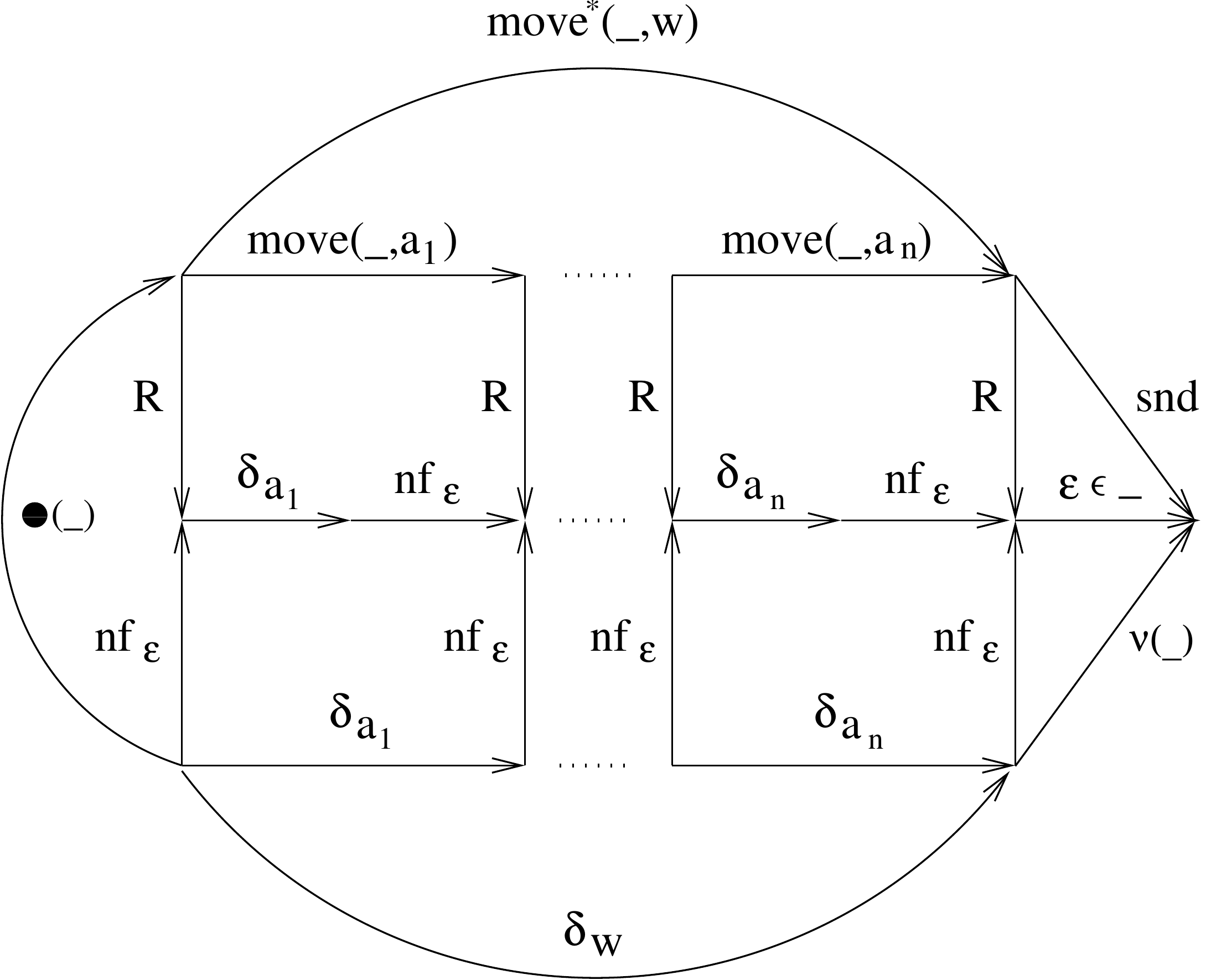}
\caption{Pointed regular expressions and Brzozowski's derivatives\label{commute}}
\end{center}
\end{figure}

\subsection{Formal proof of the commuting diagram in Figure \ref{commute}}

Part of the diagram has been already proved: the leftmost triangle, used to
relate the initial state of the two automata, is
Corollary~\ref{R-bullet}; the two triangles at the right, used to
relate the final states, just states the trivial properties that 
$\epsilon \in  R(\langle e,b\rangle)$ iff and only if $b=true$ (since
no expression in $R(e)$ is nullable), and
$\epsilon \in  \nf(e)$ if and only if $e$ is nullable (see
Remark~\ref{not_nullable}).

We start proving the upper part. We prove it for a pointed item $e$
and leave the obvious generalization to a pointed expression to the reader
(the move operation does not depend from the presence of a trailing
point, and similarly the derivative of $\epsilon$ is empty).

\begin{theorem}\label{techy}For any pointed item $e$,
  \[R(move(e,a)) = \nf(\der{a}{R(e)})\]
\end{theorem}
\begin{proof}
By induction on the structure of $e$:
\begin{itemize}
\item the cases $\emptyset$, $\epsilon$, $a$ and $b$ are trivial
\item if $e=\bullet a$ then $move(\bullet a,a) = \langle a,true \rangle$
and  $R\langle a,true \rangle = \{\epsilon\}$. On the other side, 
$\nf(\der{a}{R(\bullet a)} = \nf(\der{a}{\{a\}}) = \nf(\{\epsilon\}) = {\epsilon}$.
\item if $e=e_1+e_2$, then
\[
\begin{array}{l}
R(move(e_1+e_2,a)) =\\
\quad= R(move(e_1,a) \oplus move(e_2,a))\\
\quad= R(move(e_1,a)) \cup R(move(e_2,a))\\
\quad= \nf(\der{a}{R(e_1)}) \cup \nf(\der{a}{R(e_2)})\\
\quad= \nf(\der{a}{R(e_1+e_2)})\\
\end{array}
\]
\item let $e=e_1e_2$, and let us suppose that $move(e_1,a) =
\langle e_1', \false \rangle$ and thus $R(move(e_1,a) = R(e_1')$
and $\nul{\der{a}{R(e_1)}} = \false$. Then
\[
\begin{array}{l}
R(move(e_1e_2,a)) =\\
\quad= R(move(e_1,a) \odot move(e_2,a))\\
\quad= R(move(e_1,a))|move(e2,a)| \cup R(move(e_2,a))\\
\quad= \nf(\der{a}{R(e_1)})|e_2| \cup \nf(\der{a}{R(e_2)})\\
\quad= \nf(\der{a}{R(e_1)}|e_2| \cup \der{a}{R(e_2)})\\
\quad= \nf(\der{a}{R(e_1)|e_2|} \cup \der{a}{R(e_2)})\\
\quad= \nf(\der{a}{R(e_1)|e_2| \cup R(e_2)})\\
\quad= \nf(\der{a}{R(e_1e_2)})\\
\end{array}
\]
If $move(e_1,a) = \langle e_1', \true \rangle$ then
$R(move(e_1,a)) = R(e_1') \cup {\epsilon} = \nf(\der{a}{R(e_1)}$. In particular
$R(e_1') = \dnf(\der{a}{R(e_1)}$ and $\nul{\der{a}{R(e_1)}} = \true$.
We have then:
\[
\begin{array}{l}
R(move(e_1e_2,a)) =\\
\quad= R(move(e_1,a) \odot move(e_2,a))\\
\quad= R(e_1')|move(e_2,a)| \cup \nf(|move(e_2,a)|) \cup R(move(e_2,a))\\
\quad= R(e_1')|e_2| \cup \nf(|e_2|) \cup R(move(e_2,a))\\
\quad= \dnf(\der{a}{R(e_1)})|e_2| \cup \nf(|e_2|) \cup \nf(\der{a}{R(e_2)})\\ 
\quad= \nf(\der{a}{R(e_1)}|e_2|) \cup \nf(\der{a}{R(e_2)})\\ 
\quad= \nf(\der{a}{R(e_1)}|e_2| \cup \der{a}{R(e_2)})\\ 
\quad= \nf(\der{a}{R(e_1)|e_2|} \cup \der{a}{R(e_2)})\\
\quad= \nf(\der{a}{R(e_1)|e_2| \cup R(e_2)})\\
\quad= \nf(\der{a}{R(e_1e_2)})\\
\end{array}
\]
\item let $e=e_1^*$, and let us suppose that $move(e_1,a) =
\langle e_1', \false \rangle$. Thus $\epsilon \not \in \nf(\der{a}{R(e_1)})$.
Then
\[
\begin{array}{l}
R(move(e_1^*,a)) =\\
\quad = R(move(e_1,a)^\varoast) \\
\quad= R(e_1')|e_1^*|\\
\quad= \nf(\der{a}{R(e_1)})|e_1^*| \\
\quad= \nf(\der{a}{R(e_1)}|e_1^*|) \\
\quad= \nf(\der{a}{R(e_1)|e_1^*|)})\\
\quad= \nf(\der{a}{R(e_1^*)})\\
\end{array}
\]
If $move(e_1,a) = \langle e_1', \true \rangle$ then
$R(move(e_1,a)) = R(e_1') \cup {\epsilon} = \nf(\der{a}{R(e_1)}$.
In particular
$R(e_1') = \dnf(\der{a}{R(e_1)}$ and $\nul{\der{a}(R(e_1))} = \true$ since
$\epsilon \in \nf(\der{a}{R(e_1)}$.
We have then:
\[
\begin{array}{l}
R(move(e_1^*,a)) =\\
\quad = R(move(e_1,a)^\varoast) \\
\quad= R(e_1')|e_1^*| \cup \nf(|e_1^*|)\\
\quad= \dnf(\der{a}{R(e_1)})|e_1^*| \cup \nf(|e_1^*|)\\
\quad= \nf(\der{a}{R(e_1))|e_1^*|}\\
\quad= \nf(\der{a}{R(e_1)|e_1^*|)}\\
\quad= \nf(\der{a}{R(e_1^*)})\\
\end{array}
\]
\end{itemize}
\end{proof}


We pass now to prove the lower part of the diagram in Figure~\ref{commute},
namely that for any regular expression $e$, 
\[\nf(\der{a}{e}) = \nf(\der{a}{\nf(e)}) \] 
Since however, 
$\nf(\der{a}{\nf(e)}) = \nf(\der{a}{\dnf(e)})$ (the derivative of $\epsilon$
is empty), this is equivalent to prove the following
result.



\begin{theorem}\label{techw}
$
\nf(\der{a}{e}) = \nf(\der{a}{\dnf(e)})
$
\end{theorem}
\begin{proof}
The proof is by induction on $e$.
Any induction hypothesis over a regular
expression $e_1$ can be strengthened to
$\nf(\der{a}{e_1} e_2) = \nf(\der{a}{\dnf(e_1)} e_2)$ for all $e_2$ since
$$\begin{array}{l}
\nf(\der{a}{e_1}e_2) \\
\quad = \nf(\der{a}{e_1})e_2 \cup (\nf(e_2) \mbox{ if } \nul{\der{a}{e_1}}\\
\quad = \nf(\der{a}{\dnf(e_1)})e_2 \cup (\nf(e_2) \mbox{ if } \nul{\der{a}{\dnf(e_1)}}\\
\quad = \nf(\der{a}{\dnf(e_1)}e_2)\\
\end{array}$$
(observe that $\nul{\der{a}{e_1}} = \nul{\der{a}{\dnf(e_1)}}$ since the languages denoted by $\der{a}{e_1}$ and $\der{a}{\dnf(e_1)}$ are equal).\\
We must consider the following cases.
\begin{itemize}
\item If $e$ is $\epsilon$, $\emptyset$ or a symbol $b$ different from
$a$ then both sides of the equation are empty
\item If $e$ is $a$,
 $\nf(\der{a}{a}) = \nf(\epsilon) = \{\epsilon\} = \nf(\der{a}{\{a\}}) =
  \nf(\der{a}{\dnf(a)})$
\item If $e$ is $e_1 + e_2$,
$$\begin{array}{l}
 \nf(\der{a}{e_1 + e_2}) =\\
\quad = \nf(\der{a}{e_1} + \der{a}{e_2})\\
\quad = \nf(\der{a}{e_1}) \cup \nf(\der{a}{e_2})\\
\quad = \nf(\der{a}{\dnf(e_1)}) \cup \nf(\der{a}{\dnf(e_2)})\\
\quad = \nf(\der{a}{\dnf(e_1) \cup \dnf(e_2)})\\
\quad = \nf(\der{a}{\dnf(e_1 + e_2)})
\end{array}$$
\item If $e$ is $e_1 e_2$ and $\nul{e_1} = \false$,
$$\begin{array}{l}
 \nf(\der{a}{e_1 e_2}) =
\nf(\der{a}{e_1}e_2)
= \nf(\der{a}{\dnf(e_1)}e_2) =\\
\quad = \nf(\der{a}{\dnf(e_1) e_2})
= \nf(\der{a}{\dnf(e_1e_2)})
\end{array}$$
\item If $e$ is $e_1 e_2$ and $\nul{e_1} = \true$,
$$\begin{array}{l}
 \nf(\der{a}{e_1 e_2}) =\\
\quad = \nf(\der{a}{e_1}e_2) \cup \nf(\der{a}{e_2}) \\
\quad = \nf(\der{a}{\dnf(e_1)}e_2) \cup \nf(\der{a}{\dnf(e_2)})\\
\quad = \nf(\der{a}{\dnf(e_1)e_2 \cup \dnf(e_2)})\\
\quad = \nf(\der{a}{\dnf(e_1e_2)})
\end{array}$$
\item If $e$ is $e_1^*$,
$$\begin{array}{l}
 \nf(\der{a}{e_1^*})
= \nf(\der{a}{e_1}e_1^*)
= \nf(\der{a}{\dnf(e_1)}e_1^*) = \\
\quad = \nf(\der{a}{\dnf(e_1)e_1^*})
= \nf(\der{a}{\dnf(e_1^*)})
\end{array}$$
\end{itemize}
\end{proof}

\begin{lemma}
$ R(e) = \nf(R(e)) $
\end{lemma}
\begin{proof}
We proceed by induction over $e$:
\begin{itemize}
\item 
$R(\emptyset) = \emptyset = \nf(\emptyset) = \nf(R(\emptyset))$
\item
$R(\epsilon) = \emptyset = \nf(\emptyset) = \nf(R(\epsilon))$
\item $R(a) = \emptyset = \nf(\emptyset) = \nf(R(a))$
\item $R(\bullet a) = \{a\} = \nf(\{a\}) = \nf(R(a))$
\item $R(e_1 + e_2) = R(e_1) \cup R(e_2)
      = \nf(R(e_1)) \cup \nf(R(e_2))
      = \nf(R(e_1) \cup R(e_2))
      = \nf(R(e_1 + e_2))$
\item $R(e_1e_2) = R(e_1)|e_2| \cup R(e_2)
 = \nf(R(e_1))|e_2| \cup \nf(R(e_2))
 = \nf(R(e_1)|e_2|) \cup \nf(R(e_2))
 = \nf(R(e_1)|e_2| \cup R(e_2)) 
 = \nf(R(e_1e_2))
$
\item $R(e^*) = R(e)|e|^*
 =  \nf(R(e))|e|^*
 =  \nf(R(e)|e|^*)
 =  \nf(R(e^*))$
\end{itemize}
\end{proof}

We are now ready to prove the commutation of the outermost diagram.

\begin{theorem}For any pointed item $e$,
  \[R(move^*(e,w)) = \nf(\der{w}{R(e)})\]
\end{theorem}
\begin{proof}
The proof is by induction on the structure of $w$.
In the base case, $R(move^*(e,\epsilon)) = R(e) = \nf(R(e)) =
\nf(\der{\epsilon}{R(e)})$. In the inductive step, by Theorem~\ref{techw},
$$\begin{array}{l}
R(move^*(e,aw)) =\\
\quad = R(move^*(move(e,a),w)\\
\quad = \nf(\der{w}{R(move(e,a)})\\
\quad = \nf(\der{w}{\nf(\der{a}{R(e)})})\\
\quad = \nf(\der{w}{\der{a}{R(e)}})\\
\quad = \nf(\der{aw}{R(e)})\\
\end{array}
$$
\end{proof}
\begin{corollary}\label{corollary1}
For any regular expression $e$,
$$ R(move^*(\bullet e,w)) = \nf(\der{w}{e})$$
\end{corollary}
\begin{proof}
$$
\!
R(move^*(\bullet e,w))
\!=\! \nf(\der{w}{R(\bullet e)}
\!=\! \nf(\der{w}{\nf(e)}
\!=\! \nf(\der{w}{e})
$$
\end{proof}

Another important consequence of Lemmas \ref{techy} and \ref{techw}
is that $R$ and $\nf$ are admissible relations (respectively, over
pres and over derivatives).

\begin{theorem}\label{admissible1}
$kn(R(\cdot))$ (the kernel of $R(\cdot)$) is an admissible
equivalence relation over pres.
\end{theorem}
\begin{proof}
By Lemma~\ref{techx} we derive that for all pres $e_1,e_2$,
if $R(e_1) = R(e_2)$ then $L_p(e_1) = L_p(e_2)$.
We also need to prove that for all pres $e_1,e_2$ and all symbol $a$,
if $R(e_1) = R(e_2)$ then
$R(move(e_1,a)) = R(move(e_2,a))$. By Theorem~\ref{techy}
$$\begin{array}{l}
R(move(e_1,a)) =
\nf(\der{a}{R(e_1)}
= \nf(\der{a}{R(e_2)} =\\
\quad = R(move(e_2,a))
\end{array}
$$
\end{proof}

\begin{theorem}\label{admissible2}
$kn(\nf(e))$ is an admissible equivalence relation over regular
expressions
\end{theorem}
\begin{proof}
By Lemma~\ref{techz} we derive that for all regular expressions $e_1,e_2$,
if $\nf(e_1) = \nf(e_2)$ then $L(e_1) = L(e_2)$.
We also need to prove that for all regular expressions $e_1,e_2$
and all symbol $a$, if $\nf(e_1) = \nf(e_2)$ then
$\nf(\der{a}{e_1}) = \nf(\der{a}{e_2})$.\\ By Theorem~\ref{techw}
$$\begin{array}{l}
\nf(\der{a}{e_1}) =
\nf(\der{a}{\nf(e_1)}
= \nf(\der{a}{\nf(e_2)} =\\
\quad = \nf(\der{a}{e_2})
\end{array}
$$
\end{proof}

\begin{theorem}\label{thesame}~\\
For each regular expression $e$, let
$D_e^\bullet = (Q^\bullet,\Sigma,\bullet e,t^\bullet,F^\bullet)$
be the automaton
for $e$ built according to Definition~\ref{algorithmic1} and let
$D_e^\delta = (Q^\delta,\Sigma,e,t^\delta,F^\delta)$ the automaton for
$e$ obtained with derivatives. Let $kn(R)$ and $kn(\nf{})$ be the kernels of $R$
and $\nf{}$ respectively. Then
$D_e^\bullet/_{kn(R)} = D_e^\delta/_{kn(\nf{})}$.
\end{theorem}
\begin{proof}
The results holds by commutation of Figure~\ref{commute}, that is granted
by the previous results, in particular by Corollary~\ref{corollary1},
Theorem~\ref{admissible1}, Theorem~\ref{admissible2}, and the commutation
of the triangles relative to the initial and final states.
\end{proof}

Theorem~\ref{thesame} relates our finite automata with the infinite states ones
obtained via Brzozowski's derivatives before quotienting the automata states
by means of $ACI$ to make them finite. The following easy lemma shows that
$kn(\nf)$ is an equivalence relation finer than $ACI$ and thus
Theorem~\ref{thesame} also holds for the standard finite Brzozowski's automata
since we can quotient with $ACI$ first.
\begin{lemma}
Let $e_1$ and $e_2$ be regular expressions.
If $e_1 =_{ACI} e_2$ then $\nf(e_1) = \nf(e_2)$.
\end{lemma}

\section{Merging}\label{merging}
By Theorem~\ref{theo:broadcast}, $\Lp{\bullet e} = \Lp e \cup \Le{|e|}$.
A more syntactic way to look at this result is to observe that 
$\bullet(e)$ can be obtained by ``merging'' together the points in $e$ and 
$\bullet (|e|)$, and that the language defined by merging two pointed
expressions $e_1$ and $e_2$ is just the union of the two languages 
$\Lp{e_1}$ and $\Lp{e_2}$. The merging operation, that we shall denote
with a $\dagger$, does also provide the relation between deterministic 
and nondeterministic automata where, as in Watson \cite{Watson01, Watson02}, 
we may label states with expressions with a single point (for lack of space,
we shall not explicitly address the latter issue in this paper, that 
is however a simple consequence of Theorem~\ref{theo:move_dag}).
Finally, the merging operation will allow us to explain why the
technique of pointed expressions cannot be (naively) generalized 
to intersection and complement (see Section~\ref{sec:intersection}).


\begin{definition}
Let $e_1$ and $e_2$ be two items on the same carrier $|e|$. The merge of $e_1$
and $e_2$ is defined by the following rules by recursion over the structure of
$e$:
\[
\begin{array}{rcl}
\emptyset \dagger \emptyset & = & \emptyset \\
\epsilon \dagger \epsilon & = & \epsilon \\
a \dagger a & = & a \\
\bullet a \dagger a & = & \bullet a \\
a \dagger \bullet a & = & \bullet a \\
\bullet a \dagger \bullet a & = & \bullet a \\
(e^1_1 + e^1_2) \dagger (e^2_1 + e^2_2) & = & (e^1_1 \dagger e^2_1) + (e^1_2 \dagger e^2_2) \\
(e^1_1e^1_2) \dagger (e^2_1e^2_2) & = & (e^1_1 \dagger e^2_1)(e^1_2 \dagger e^2_2) \\
e_1^* \dagger e_2^* & = & (e_1 \dagger e_2)^*
\end{array}
\]
The definition is extended to pres as follows:
$$\langle e_1,b_1 \rangle \dagger \langle e_2,b_2 \rangle = \langle e_1 \dagger e_2, b_1 \vee b_2 \rangle$$
\end{definition}

\begin{theorem}$\dagger$ is commutative, associative and idempotent
\end{theorem}
\begin{proof}
Trivial by induction over the structure of the carrier of the arguments.
\end{proof}

\begin{theorem}
$L_p(e_1 \dagger e_2) = L_p(e_1) \cup L_p(e_2)$
\end{theorem}
\begin{proof}
Trivial by induction on the common carrier of the items of $e_1$ and $e_2$.
\end{proof}

All the constructions we presented so far commute with the merge operation.
Since merging essentially corresponds to the subset construction over automata,
the following theorems constitute the proof of correctness of the subset
construction.

\begin{theorem}
$(e^1_1 \dagger e^2_1) \oplus (e^1_2 \dagger e^2_2) =
  (e^1_1 \oplus e^1_2) \dagger (e^2_1 \oplus e^2_2)$
\end{theorem}
\begin{proof}
Trivial by expansion of definitions.
\end{proof}

\begin{theorem}~

\begin{enumerate}
\item for $e_1$ and $e_2$ items on the same carrier,
$$\bullet(e_1 \dagger e_2) = \bullet(e_1) \dagger \langle e_2,\false \rangle$$
\item for $e_1$ and $e_2$ pres on the same carrier,
$$\bullet(e_1 \dagger e_2) = \bullet(e_1) \dagger e_2$$
\item
$(e^1_1 \dagger e^2_1) \odot (e^1_2 \dagger e^2_2) =
  (e^1_1 \odot e^1_2) \dagger (e^2_1 \odot e^2_2)$
\end{enumerate}
\begin{corollary}
$$\bullet(e_1 \dagger e_2) = e_1 \dagger \bullet(e_2) = \bullet(e_1) \dagger \bullet(e_2)$$
\end{corollary}
\begin{proof}[of the corollary]
The corollary is a simple consequence of commutativity of $\dagger$ and
idempotence of $\bullet(\cdot)$:
$$\bullet(e_1 \dagger e_2) = \bullet(e_2 \dagger e_1) = \bullet(e_2) \dagger e_1 = e_1 \dagger \bullet(e_2)$$
$$\bullet(e_1 \dagger e_2) =
\bullet(\bullet(e_1 \dagger e_2)) =
\bullet(\bullet(e_1) \dagger e_2) =
\bullet(e_1) \dagger \bullet(e_2)
$$
\end{proof}
\end{theorem}
\begin{proof}[of 1.]
We first prove $\bullet(e_1 \dagger e_2) = \bullet(e_1) \dagger \langle e_2, \false \rangle$ by
induction over the structure of the common carrier of $e_1$ and $e_2$, assuming
that 3. holds on terms whose carrier is structurally smaller than $e$.

\begin{itemize}
\item If $|e_1|$ is $\emptyset$, $\epsilon$, $a$, $\bullet a$ then trivial
\item If $e_1$ is $e^1_1 + e^2_1$ and $e_2$ is $e^1_2 + e^2_2$:
\[\begin{array}{l}
\bullet((e^1_1 + e^2_1) \dagger (e^1_2 + e^2_2)) = \\
\quad = \bullet((e^1_1 \dagger e^1_2) + (e^2_1 \dagger e^2_2)) \\
\quad = \bullet(e^1_1 \dagger e^1_2) \oplus \bullet(e^2_1 \dagger e^2_2) \\
\quad = (\bullet(e^1_1) \dagger \langle e^1_2, \false \rangle) \oplus (\bullet(e^2_1) \dagger \langle e^2_2, \false \rangle) \\
\quad = (\bullet(e^1_1) \oplus \bullet(e^2_1)) \dagger (\langle e^1_2,\false \rangle \oplus \langle e^2_2, \false \rangle) \\
\quad = \bullet(e^1_1 + e^2_1) \dagger \langle e^1_2 + e^2_2, \false \rangle \\
\end{array}\]
\item If $e_1$ is $e^1_1 e^2_1$ and $e_2$ is $e^1_2 e^2_2$ then, using 3.
on items whose carrier is structurally smaller than $|e_1|$,
\[\begin{array}{l}
\bullet((e^1_1 e^2_1) \dagger (e^1_2 e^2_2)) = \\
\quad = \bullet((e^1_1 \dagger e^1_2) (e^2_1 \dagger e^2_2)) \\
\quad = \bullet(e^1_1 \dagger e^1_2) \odot \langle e^2_1 \dagger e^2_2, \false \rangle \\
\quad = (\bullet(e^1_1) \dagger \langle e^1_2,\false\rangle) \odot (\langle e^2_1, \false \rangle \dagger \langle e^2_2, \false \rangle) \\
\quad = (\bullet(e^1_1) \odot \langle e^2_1, \false \rangle) \dagger (\langle e^1_2, \false \rangle \odot \langle e^2_2, \false \rangle)) \\
\quad = \bullet(e^1_1 e^2_1) \dagger \langle e^1_2 e^2_2, \false \rangle \\
\end{array}\]
\item If $e_1$ is ${e^1_1}^*$ and $e_2$ is ${e^1_2}^*$,
let $\bullet(e^1_1 \dagger e^1_2) = \langle e',b' \rangle$ and
$\bullet(e^1_1) = \langle e'',b'' \rangle$. By induction hypothesis,
$\langle e',b' \rangle = \bullet(e^1_1 \dagger e^1_2) =
 \bullet(e^1_1) \dagger \langle e^1_2, \false \rangle = \langle e'',b'' \rangle \dagger \langle e^1_2, \false \rangle$
Then\\
$$
\begin{array}{l}
\bullet({e^1_1}^* \dagger {e^1_2}^*) =
\bullet ((e^1_1 \dagger e^1_2)^*)
= \langle e'^*,\true \rangle = \\
\quad = \langle e''^*, \true \rangle \dagger \langle {e^1_2}^*, \false \rangle
= \bullet({e^1_1}^*) \dagger \langle {e^1_2}^*, \false \rangle
\end{array}
$$
\end{itemize}
\end{proof}
\begin{proof}[Of 2.]
Let $\langle e'^i_j,b'^i_j \rangle = e^i_j$.
By definition of $\dagger$, we have
$$e^1_1 \dagger e^2_1 = \langle e'^1_1 \dagger e'^2_1, b'^1_1 \vee b'^2_1 \rangle$$
For all $b$ and $e$, let $\bullet_b(e) := \begin{cases} e & \mbox{ if $b = \false$ }\\
                                                \bullet(e) & \mbox{ otherwise } \end{cases}$\\
Thus for all $e'_1,e'_2,b'_1,b'_2$, letting
$\langle e''_2,b''_2 \rangle := \bullet_{b'_1} (\langle e'_2,b'_2\rangle)$, the following holds:
$$\langle e'_1, b'_1 \rangle \odot \langle e'_2, b'_2 \rangle = \langle e'_1 e''_2, b'_2 \vee b''_2 \rangle$$
Let $\langle e''^i_2,b''^i_2 \rangle := \bullet_{b'^i_1}(e^i_2)$. By property 1. we have: $$
\langle e''^1_2 \dagger e''^2_2, b''^1_2 \vee b''^2_2 \rangle =
\bullet_{b'^1_1}(e^1_2) \dagger \bullet_{b'^2_1}(e^2_2) =
\bullet_{b'^1_1 \vee b'^2_1}(e^1_2 \dagger e^2_2)$$

Thus
\[\begin{array}{l}
(e^1_1 \dagger e^2_1) \odot (e^1_2 \dagger e^2_2) = \\
\quad = \langle e'^1_1 \dagger e'^2_1, b'^1_1 \vee b'^2_1 \rangle \odot
        (e^1_2 \dagger e^2_2) \\
\quad = \langle (e'^1_1 \dagger e'^2_1)(e''^1_2 \dagger e''^2_2),
       b'^1_2 \vee b'^2_2 \vee b''^1_2 \vee b''^2_2 \rangle \\
\quad = \langle (e'^1_1 e''^1_2) \dagger (e'^2_1 e''^2_2),
       (b'^1_2 \vee b''^1_2) \vee (b'^2_2 \vee b''^2_2) \rangle \\
\quad = \langle e'^1_1 e''^1_2, b'^1_2 \vee b''^1_2 \rangle \dagger
        \langle e'^2_1 e''^2_2, b'^2_2 \vee b''^2_2 \rangle \\
\quad = (e^1_1 \odot e^1_2) \dagger (e^2_1 \odot e^2_2) \\
\end{array}\]
\end{proof}

\begin{theorem}
$
(e_1 \dagger e_2)^\varoast = e_1^\varoast \dagger e_2^\varoast
$
\end{theorem}
\begin{proof}
Let $e_1 = \langle e_1^1,b_1 \rangle$ and $e_2 = \langle e_2^1,b_2 \rangle$.
Thus
$$
(\langle e_1^1,b_1 \rangle \dagger \langle e_2^1,b_2 \rangle)^\varoast
= \langle e_1^1 \dagger e_2^1, b_1 \vee b_2 \rangle^\varoast
$$

Let define $e'$, $e_1'$ and $e_2'$ by cases on $b_1$ and $b_2$ with
the property that $e' = e_1' \dagger e_2'$:
\begin{itemize}
 \item If $b_1 = b_2 = \false$ then let $e_i' = e_i^1$ and
  $e' = e_1^1 \dagger e_2^1$. Obviously $e' = e_1' \dagger e_2'$.
 \item If $b_1 = \true$ and $b_2 = \false$ then
   let $\bullet (e_1^1) = \langle e_1', b_1' \rangle$, let
   $e_2' = e^1_2$ and let $\bullet (e_1^1 \dagger e_2^1) =
   \bullet (e_1^1) \dagger \langle e_2^1, \false \rangle =
   \langle e',b' \rangle$. Hence $e_1' \dagger e_2^1 = e_1' \dagger e_2' = e'$.
 \item The case $b_1 = \false$ and $b_2 = \true$ is handled dually to the
   previous one.
 \item If $b_1 = \true$ and $b_2 = \true$ then
   let $\bullet (e_i^1) = \langle e_i', b_i' \rangle$
   and let $\bullet (e_1^1 \dagger e_2^1) =
   \bullet (e_1^1) \dagger \bullet (e_2^1) =
   \langle e',b' \rangle$. Hence $e_1' \dagger e_2' = e'$.
\end{itemize}

In all cases,
$$\begin{array}{l}
 \langle e_1^1 \dagger e_2^1, b_1 \vee b_2 \rangle^\varoast
= \langle {e'}^*, b_1 \vee b_2 \rangle
= \langle (e_1' \dagger e_2')^*, b_1 \vee b_2 \rangle =\\
\quad = \langle {e_1'}^* \dagger {e_2'}^*, b_1 \vee b_2 \rangle
= \langle {e_1'}^*, b_1 \rangle \dagger
  \langle {e_2', b_2}^* \rangle \\
\quad = \langle e_1^1, b_1 \rangle^\varoast \dagger
  \langle e_2^1, b_2 \rangle^\varoast
\end{array}$$
\end{proof}

\begin{theorem}
\label{theo:move_dag}
$\quad move(e_1\dagger e_2,a) = move(e_1,a) \dagger  move(e_2,a)$
\end{theorem}
\begin{proof}
The proof is by induction on the structure of $e$.
\begin{itemize}
\item the cases $\emptyset$, $\epsilon$ and $b \neq a$ are trivial by
      computation
\item the case $a$ has four sub-cases: if $e_1$ and $e_2$ are both $a$,
      then $move(a \dagger a,a) = \langle \emptyset, \false \rangle =
      move(a,a) \dagger  move(a,a)$;
      otherwise
      at least one in $e_1$ or $e_2$ is $\bullet a$ and
      $move(e_1 \dagger e_2,a) = move(\bullet a,a) = \langle a,true \rangle =
      move(e_1,a) \dagger  move(e_2,a)$
\item if $e$ is $e^1 + e^2$ then
\[
\begin{array}{l}
move((e^1_1 + e^2_1)  \dagger (e^1_2 + e^2_2),a) = \\
\; = move((e^1_1 \dagger e^1_2) + (e^2_1 \dagger e^2_2),a) \\
\; = move(e^1_1 \dagger e^1_2,a) \oplus move(e^2_1 \dagger e^2_2,a) \\
\; = (move(e^1_1,a) \dagger move(e^1_2,a)) \oplus
        (move(e^2_1,a) \dagger move(e^2_2,a)) \\
\; = (move(e^1_1,a) \oplus move(e^2_1,a)) \dagger
        (move(e^1_2,a) \oplus move(e^2_2,a)) \\
\; = move(e^1_1+e^2_1,a) \dagger
        move(e^1_2+e^2_2,a) \\
\end{array}
\]
\item if $e$ is $e^1 e^2$ then
\[
\begin{array}{l}
move((e^1_1 e^2_1)  \dagger (e^1_2 e^2_2),a) = \\
\; = move((e^1_1 \dagger e^1_2) (e^2_1 \dagger e^2_2),a) \\
\; = move(e^1_1 \dagger e^1_2,a) \odot move(e^2_1 \dagger e^2_2,a) \\
\; = (move(e^1_1,a) \dagger move(e^1_2,a)) \odot
        (move(e^2_1,a) \dagger move(e^2_2,a)) \\
\; = (move(e^1_1,a) \odot move(e^2_1,a)) \dagger
        (move(e^1_2,a) \odot move(e^2_2,a)) \\
\; = move(e^1_1 e^2_1,a) \dagger
        move(e^1_2 e^2_2,a) \\
\end{array}
\]
\item if $e$ is ${e^1}^*$ then
\[
\begin{array}{l}
move({e_1^1}^* \dagger {e_2^1}^*)
= move((e_1^1 \dagger e_2^1)^*) = \\
\quad = move(e_1^1 \dagger e_2^1)^\varoast
= (move(e_1^1) \dagger move(e_2^1))^\varoast\\
\quad = move(e_1^1)^\varoast \dagger move(e_2^1)^\varoast
= move({e_1^1}^*) \dagger move({e_2^1}^*)\\
\end{array}
\]
      
\end{itemize}
\end{proof}

\subsection{Intersection and complement}
\label{sec:intersection}
Pointed expressions cannot be generalized in a trivial way to
the operations of intersection and complement. Suppose to 
extend the definition of the language in the obvious way, letting
$\Lp{e_1 \cap e_2} = \Lp{e_1} \cap \Lp{e_2}$ and 
$\Lp{\neg e} = \overline{\Lp{e}}$. The problem is that merging
is no longer additive, and Theorem~\ref{theo:broadcast} does not
hold any more.
For instance, consider the two expressions $e_1 = \bullet a \cap a$ and
$e_2 = a \cap \bullet a$. Clearly $\Lp{e_1} = \Lp{e_2} = \emptyset$, but
$\Lp {e_1 \dag e_2} = \Lp{\bullet a \cap \bullet a} = \{a\}$.
To better understand the problem, 
let $e = (\bullet ba \cap \bullet a)| \bullet b$, and let us 
consider the result of $move(e^*,b)$. 
Since $move(e,b) = \langle (b\bullet a \cap a)| b), true \rangle$, we should
broadcast a new point inside $(b\bullet a \cap a)| b)$, 
hence $move(e^*,b) = (\bullet b\bullet a \cap \bullet a)| \bullet b)^*$,
that is obviously wrong.

The problems in extending the technique to intersection and complement are
not due to some easily avoidable deficiency of the approach but 
have a deep theoretical reason: indeed, even if these operators do not increase
the expressive power of regular expressions they can have a drastic
impact on succinctness, making them much harder to handle.
For instance it is well known that expressions with
complements can provide descriptions of certain languages which are
non-elementary more compact than standard regular expression \cite{MeyerS72}.
Gelade \cite{Gelade10} has recently proved that for any natural 
number $n$ there exists a regular expression with intersection
of size $\mathcal{O}(n)$ such that any DFA accepting its language has
a double-exponential size, i.e. it contains at least $2^{2^n}$ states
(see also \cite{GruberH08}).
Hence, marking positions with points is not enough, just because we 
would not have enough states. 

Since the problem is due to a loss of information
during merging, we are currently investigating the possibility
to exploit {\em colored} points. An important goal of this
approach would be to provide simple, completely
syntactic explanations for space bounds of different classes
of languages.

\section{Conclusions}
We introduced in this paper the notion of pointed regular
expression, investigated its main properties, and its 
relation with Brzozowski's derivatives. 
Points are used to mark the positions inside the regular
expression which have been reached after reading some prefix of
the input string, and where the processing
of the remaining string should start. In particular,
each pointed expression has a clear semantics. Since
each pointed expression for $e$ represents a state of 
the {\em deterministic} automaton associated
with $e$, this means we may associate a semantics to each
state in terms of the specification $e$ 
and not of the behaviour of the automaton. 
This allows a {\em direct}, {\em intuitive} and 
{\em easily verifiable} construction of the deterministic 
automaton for $e$. 

A major advantage of pointed expressions 
is from the didactical point of view. Relying on 
an electronic device, it is a real pleasure to see
points moving inside the regular expression in response to
an input symbol. Students immediately grasp the idea, and 
are able to manually build the automata, and to understand the
meaning of its states, after a single lesson. Moreover, if you 
have a really short time, you can altogether skip the notion of 
nondeterministic automata.

Regular expression received a renewed interest in recent 
years, mostly due to their use in XML-languages. 
Pointed expressions seem to open a huge range of novel 
perspectives and original approaches in the field, starting
from the {\em challenging} generalization of the approach to
different operators such as counting, intersection, 
and interleaving (e.g. exploiting colors for points, 
see Section \ref{sec:intersection}). A large amount
of research has been recently devoted to the so called
succinteness problem, namely the investigation
of the descriptional complexity of regular languages
(see e.g. \cite{Gelade10, GruberH08, HolzerK09}). Since,
as observed in Example\ref{ex:compact}, pointed expression
can provide a more compact description for regular languages
than traditional regular expression, it looks interesting to better 
investigated this issue (that seems to be related to the so called
star-height \cite{eggan63} of the language).


It could also be worth to investigate variants of the
notion of pointed expression, allowing different
positioning of points inside the expressions. 
Merging must be better investigated, and the whole equational
theory of pointed expressions, both with different and
(especially) fixed carriers must be entirely developed.

As explained in the introduction, the notion of pointed 
expression was suggested by an attempt of formalizing the theory
of regular languages by means of an interactive prover.
This testify the relevance of the choice of good data structures 
not just for the design of algorithms but also for the formal 
investigation of a given field, and is a 
nice example of the kind of interesting feedback one may expect 
by the interplay with automated devices for proof development. 


\bibliographystyle{ieeetr}
\bibliography{../BIBTEX/helm}

\begin{thebibliography}{10}

\bibitem{RS97}
G.~Rozenberg and A.~Salomaa, eds., {\em Handbook of formal languages, vol. 1:
  word, language, grammar}.
\newblock New York, NY, USA: Springer-Verlag New York, Inc., 1997.

\bibitem{EllulKSW05}
K.~Ellul, B.~Krawetz, J.~Shallit, and M.~wei Wang, ``Regular expressions: New
  results and open problems,'' {\em Journal of Automata, Languages and
  Combinatorics}, vol.~10, no.~4, pp.~407--437, 2005.

\bibitem{Brzozowski64}
J.~A. Brzozowski, ``Derivatives of regular expressions,'' {\em J. ACM},
  vol.~11, no.~4, pp.~481--494, 1964.

\bibitem{McNY60}
R.~McNaughton and H.~Yamada, ``Regular expressions and state graphs for
  automata,'' {\em Ieee Transactions On Electronic Computers}, vol.~9, no.~1,
  pp.~39--47, 1960.

\bibitem{OwensRT09}
S.~Owens, J.~H. Reppy, and A.~Turon, ``Regular-expression derivatives
  re-examined,'' {\em J. Funct. Program.}, vol.~19, no.~2, pp.~173--190, 2009.

\bibitem{BS86}
G.~Berry and R.~Sethi, ``From regular expressions to deterministic automata,''
  {\em Theor. Comput. Sci.}, vol.~48, no.~3, pp.~117--126, 1986.

\bibitem{Bruggemann-Klein93}
A.~Br{\"u}ggemann-Klein, ``Regular expressions into finite automata,'' {\em
  Theor. Comput. Sci.}, vol.~120, no.~2, pp.~197--213, 1993.

\bibitem{ChangP92}
C.-H. Chang and R.~Paige, ``From regular expressions to dfa's using compressed
  nfa's,'' in {\em Combinatorial Pattern Matching, Third Annual Symposium, CPM
  92, Tucson, Arizona, USA, April 29 - May 1, 1992, Proceedings}, vol.~644 of
  {\em Lecture Notes in Computer Science}, pp.~90--110, Springer, 1992.

\bibitem{Kleene56}
S.~C. Kleene, ``Representation of events in nerve nets and finite automata,''
  in {\em Automata Studies} (C.~E. Shannon and J.~McCarthy, eds.), pp.~3--42,
  Princeton University Press, 1956.

\bibitem{Watson01}
B.~W. Watson, ``A taxonomy of algorithms for constructing minimal acyclic
  deterministic finite automata,'' {\em South African Computer Journal},
  vol.~27, pp.~12--17, 2001.

\bibitem{Watson02}
B.~W. Watson, ``Directly constructing minimal dfas : combining two algorithms
  by brzozowski,'' {\em South African Computer Journal}, vol.~29, pp.~17--23,
  2002.

\bibitem{Gelade10}
W.~Gelade, ``Succinctness of regular expressions with interleaving,
  intersection and counting,'' {\em Theor. Comput. Sci.}, vol.~411, no.~31-33,
  pp.~2987--2998, 2010.

\bibitem{GruberH08}
H.~Gruber and M.~Holzer, ``Finite automata, digraph connectivity, and regular
  expression size,'' in {\em ICALP}, vol.~5126 of {\em Lecture Notes in
  Computer Science}, pp.~39--50, Springer, 2008.

\bibitem{HolzerK09}
M.~Holzer and M.~Kutrib, ``Nondeterministic finite automata - recent results on
  the descriptional and computational complexity,'' {\em Int. J. Found. Comput.
  Sci.}, vol.~20, no.~4, pp.~563--580, 2009.

\bibitem{MeyerS72}
A.~R. Meyer and L.~J. Stockmeyer, ``The equivalence problem for regular
  expressions with squaring requires exponential space,'' in {\em 13th Annual
  Symposium on Switching and Automata Theory (FOCS)}, pp.~125--129, IEEE, 1972.

\bibitem{eggan63}
L.~C. Eggan, ``Transition graphs and the star-height of regular events,'' {\em
  Michigan Mathematical Journal}, vol.~10, no.~4, pp.~385--397, 1963.

\end{thebibliography}

\end{document}